\documentclass[11pt,twoside]{article}
\usepackage{fancyhdr}
\usepackage[colorlinks,citecolor=blue,urlcolor=blue,linkcolor=blue,bookmarks=false]{hyperref}
\usepackage{amsfonts,epsfig,graphicx}
\usepackage{afterpage}
\usepackage{comment}
\usepackage{amsmath,amssymb,amsthm} 
\usepackage{fullpage}
\usepackage[T1]{fontenc} 
\usepackage{epsf} 
\usepackage{graphics} 
\usepackage{amsfonts,amsmath}
\usepackage{mathtools}
\usepackage{natbib} 
\usepackage{psfrag,xspace}
\usepackage{color,etoolbox}
\usepackage{xcolor}
\usepackage{subcaption} 
\usepackage{dsfont}
\usepackage[page]{appendix}
\usepackage{thmtools}
\usepackage{soul} 

\def\ind{\perp\!\!\!\perp}

\DeclareMathOperator*{\minimize}{minimize}
\DeclareMathOperator*{\argmin}{arg\,min}
\DeclareMathOperator*{\argmax}{arg\,max}

\DeclareSymbolFont{bbold}{U}{bbold}{m}{n}
\DeclareSymbolFontAlphabet{\mathbbold}{bbold}

\newtheorem{theorem}{Theorem}
\newtheorem{lemma}{Lemma}
\newtheorem{corollary}{Corollary}
\newtheorem{proposition}{Proposition}

\newtheorem{assumption}{Assumption}

\theoremstyle{definition}
\newtheorem{model}{Model}
\newtheorem{definition}{Definition}

\theoremstyle{remark}

\declaretheorem[style=definition]{example}

\begin{document}

\def\spacingset#1{\renewcommand{\baselinestretch}%
{#1}\small\normalsize} \spacingset{1}

\raggedbottom
\allowdisplaybreaks[1]

\renewcommand\thmcontinues[1]{Continued}


\title{\vspace*{-.4in} {Stochastic interventions, sensitivity analysis, \\ and optimal transport}} \author{Alexander W. Levis$^1$, Edward
  H. Kennedy$^1$, Alec McClean$^2$, \\
  Sivaraman Balakrishnan$^1$, Larry Wasserman$^1$ \\  \\
  $^1$Department of Statistics \& Data Science, \\
  Carnegie Mellon University \\
  \\
  $^2$Division of Biostatistics, \\
  NYU Grossman School of Medicine \\ \\
  \texttt{alevis@cmu.edu};
  \texttt{edward@stat.cmu.edu};
  \texttt{hadera01@nyu.edu} \\ \texttt{siva@stat.cmu.edu}; \texttt{larry@stat.cmu.edu}
  \date{} }
    
  \maketitle
  \thispagestyle{empty}

\begin{abstract}
Recent methodological research in causal inference has focused on effects of \textit{stochastic interventions}, which assign treatment randomly, often according to subject-specific covariates. 
In this work, we demonstrate that the usual notion of stochastic interventions have a surprising property: when there is unmeasured confounding, bounds on their effects do not collapse when the policy approaches the observational regime. As an alternative, we propose to study \textit{generalized policies}, treatment rules that can depend on covariates, the natural value of treatment, and auxiliary randomness. We show that certain generalized policy formulations can resolve the ``non-collapsing'' bound issue: bounds narrow to a point when the target treatment distribution approaches that in the observed data. Moreover, drawing connections to the theory of optimal transport, we characterize generalized policies that minimize worst-case bound width in various sensitivity analysis models, as well as corresponding sharp bounds on their causal effects. These optimal policies are new, and can have a more parsimonious interpretation compared to their usual stochastic policy analogues. Finally, we develop flexible, efficient, and robust estimators for the sharp nonparametric bounds that emerge from the framework.
\end{abstract}

\bigskip


\pagebreak

\section{Introduction}
In causal inference, the bulk of research has focused on the effects of deterministic, static treatment regimes, i.e., interventions that set treatment to fixed values for all individuals in a population. For instance, the canonical average treatment effect is a contrast of two interventions that each set treatment to one specific value in a population. Dynamic (or non-static) treatment regimes represent interventions that can depend on a subject's individual-level characteristics, such as their baseline covariates. A classic example arises when studying optimal treatment regimes \citep{murphy2003,robins2004}, where the focus is on identifying covariate-dependent rules that maximize outcomes. Stochastic regimes---which may be static or dynamic---represent interventions that assign treatment randomly according to some probability distribution.

Stochastic, dynamic treatment rules are appealing in many settings for a variety of reasons. For instance, they may be more appropriate when a static, deterministic intervention is infeasible, or the exposure of interest is not directly manipulable~\citep{diaz2013}, or otherwise they may simply be preferred because they represent a smaller more realistic shift from the observational treatment process~\citep{kennedy2019}. In view of the latter, \textit{incremental propensity score interventions} were proposed by~\citet{kennedy2019} for binary treatments. These interventions transform the observational treatment process such that the odds of treatment are multiplied by some specified value. As such, individuals who would deterministically receive treatment or control are not affected by this intervention, and as a result, identification and estimation of their effects do not rely on the usual positivity assumption that all subjects have a positive probability of receiving each exposure level. More generally, interventions that transform the observational treatment distribution via exponential tilts do not require a positivity assumption \citep{diaz2020, schindl2024, rakshit2024}. These and other varieties of stochastic, dynamic treatment rules have been used for point treatments~\citep{stock1989, tian2008, dudik2014, mcclean2024b}, in longitudinal settings~\citep{murphy2001, taubman2009}, and to define new types of effects for mediation analysis~\citep{vanderweele2014, diaz2020, diaz2024}.

Dynamic treatment rules that can depend on the \textit{natural value} of treatment, that which would occur in the absence of any intervention, were initially considered in~\citet{robins2004b}. These have also been referred to as a modified treatment policies (MTPs) \citep{haneuse2013}. Such treatment rules have been increasingly studied in recent years~\citep{diaz2012, moore2012, richardson2013, young2014, diaz2023}, for example encoding a shift of the natural treatment value up or down by some deterministic amount, or up to some threshold if not already exceeding it. Under the assumption that all confounders are measured, \citet{young2014} and \citet{diaz2023} associate MTPs with more traditional stochastic interventions that only depend on covariates, by considering independent random draws from the induced conditional distribution of the treatment. As we will show, the MTP and its associated stochastic policy analogue may have substantially different effects when there is unmeasured confounding. Moreover, \citet{diaz2020} note that whether or not the converse construction is possible (associating a traditional stochastic intervention with one or more MTPs) is still an open problem, which we directly address in this paper.

In this work, we introduce a class of stochastic MTPs---which we call generalized policies---associated with an arbitrary target treatment distribution, making use of the \textit{coupling method} from probability theory~\citep{strassen1965, lindvall2002}. In doing so, we show that the treatment distribution arising from a stochastic policy depending only on covariates can be achieved by any treatment rule in a large (typically infinite) class of generalized policies. This construction resolves and clarifies the connection between stochastic interventions and MTPs that has thus far not been made clear in the literature.

When there is no unmeasured confounding, the effects of generalized policies with the same treatment distribution (given covariates) are identical, and one can opt for the policy with their preferred interpretation. When there is unmeasured confounding, however, the degree to which the effect of each intervention is partially identified varies. Drawing connections with the theory of optimal transport, we consider minimizing worst-case bound length for causal parameters under unmeasured confounding. The resulting optimal generalized policies are new and of interest in their own right, often with a favorable interpretation compared to traditional stochastic policies that draw treatment independently given covariates. For instance, when the target distribution converges to the observational treatment distribution, bounds for optimal policies collapse to a point; crucially, this does not typically occur for traditional stochastic interventions. Under various sensitivity analysis frameworks, we develop valid and sharp bounds for mean potential outcomes under these interventions. Finally, we develop flexible, nonparametric efficient estimators of the bound functionals that emerge from our analysis. 

\subsection{Related work on sensitivity analysis}
A principal contribution of our work is the development of new sensitivity analyses for generalized policies, with the goal of relaxing the unconfoundedness assumption that is routinely made in observational studies. Our work builds upon the extensive literature on sensitivity analysis for causal inference, which we briefly review here---for an excellent and more exhaustive review of this topic, see Section 2 of~\citet{nabi2024}. One family of nonparametric approaches aims to construct bounds on causal effects under the weakest possible assumptions, e.g., only assuming that outcomes are bounded. This includes seminal work by~\citet{robins1989} and~\citet{manski1990}. The resulting bounds are usually quite wide given that the assumptions are so weak, though they also prove narrower bounds under stronger assumptions including monotonicity. A second class of nonparametric sensitivity analyses imposes stronger structural assumptions. This includes approaches that limit the gap between the treatment distribution given all observed and unobserved confounders, and that given only observed confounders~\citep{rosenbaum_2002, tan2006, dorn2023}, approaches that similarly limit the gap (on the difference or ratio scale) of observed and counterfactual outcome distributions~\citep{diaz2013b, luedtke2015, bonvini2022b}, and a contamination approach which limits the proportion of the population for which unmeasured confounders are at play~\citep{bonvini2022}. A third family of analyses are semiparametric approaches in which a (unidentified) parametric model captures the unmeasured confounding mechanism in some way, such that point identification is obtained at each fixed value of the unknown sensitivity parameters~\citep{robins2000, franks2020, nabi2024}. Finally, we note that much of the work on causal sensitivity analysis has focused on mean counterfactuals for deterministic, static treatment rules, or a contrasts thereof (e.g., the average treatment effect). Further, existing methods typically are restricted to binary treatments; notable exceptions are the work of~\citet{bonvini2022b}, where sensitivity analyses are proposed for marginal structural models which can incorporate arbitrary discrete or continuous treatments, and \citet{jesson2022,baitairian2024}, where bounds are developed for continuous treatments in a specific sensitivity model. To the best of our knowledge, our work represents the first in-depth exploration of nonparametric bounds and sensitivity analyses for non-deterministic policies, and among the first to consider arbitrary treatments.

\subsection{Structure of the paper}
The remainder of this article is organized as follows. In Section~\ref{sec:generalized}, we lay out notation, give preliminary definitions and assumptions, and provide an illustrative example to motivate the ensuing proposals. In Section~\ref{sec:methods}, we introduce couplings and $Q$-policies: generalized policies that yield a particular target treatment distribution, $Q$. We then consider various sensitivity models, characterize optimal $Q$-policies that minimize worst-case bound width for counterfactual means, and provide specialized results for binary treatments. In Section~\ref{sec:estimation}, focusing on exponentially tilted target distributions, we develop flexible, robust, nonparametric efficient estimators of the bound quantities resulting from the framework. Finally, in Section~\ref{sec:discussion}, we discuss and provide concluding remarks.

\section{Generalized Policies and Motivation} \label{sec:generalized}
Consider the point treatment setting with observed data comprising $n$ independent and identically distributed copies of $O = (X, A, Y) \sim P$, where $X \in \mathcal{X} \subseteq \mathbb{R}^d$ is a set of baseline covariates, $A \in \mathcal{A} \subseteq \mathbb{R}$ is a treatment variable, and $Y \in \mathcal{Y} \subseteq \mathbb{R}$ is an outcome of interest. We will consider treatment rules that can depend on covariates, the natural value of treatment, and auxiliary randomness; formally, we let $d: \mathcal{X} \times \mathcal{A} \times \mathcal{V} \to \mathcal{A}$ be any measurable function, whereby for some auxiliary noise $V \in \mathcal{V}$ independent of observables and unobservables, the rule assigns treatment via $d(X, A, V)$. As these encompass traditional stochastic policies and MTPs, we refer to these as \textit{generalized policies}---let $\mathcal{D}$ be the set of all generalized policies. In a slight abuse of notation, in cases where $d$ depends only on $X$ or $(X, A)$ or $(X, V)$, we omit the other variable(s). It is also worth mentioning that for our theoretical development, it is sufficient to take $\mathcal{V} = [0,1]$ and $V \sim \mathrm{Unif}(0,1)$. However, it is often convenient to describe interventions using more sophisticated choices for $V$, and hence leave the notation general.

Denote by $Y(d)$ the potential outcome under the rule $d \in \mathcal{D}$, and write $Y(a)$ for any $a \in \mathcal{A}$ when $d(X, A, V) \equiv a$. We will make the following two assumptions throughout.
\begin{assumption}[Consistency]\label{ass:consistency}
    $Y(d_1) = Y(d_2)$ whenever $d_1(X, A, V) = d_2(X, A, V)$, and $Y(A) = Y$.
\end{assumption}
\begin{assumption}[Auxiliary randomness]\label{ass:aux}
    $V \ind (X, A, \{Y(a)\}_{a \in \mathcal{A}})$.
\end{assumption}

\noindent Assumption~\ref{ass:consistency} requires that there is only one version of treatment and that there is no interference, and is used to link observed data to potential outcomes. Assumption~\ref{ass:aux} is required for our later results on sensitivity analysis, though is easily enforced in applications through random number generation. It can be relaxed without loss of generality to $V \ind (A, \{Y(a)\}_{a \in \mathcal{A}}) \mid X$, but we use the stronger form for simplicity.

Let $\Pi(a \mid X) = P[A \leq a \mid X]$, for $a \in \mathcal{\mathbb{R}}$, denote the conditional cumulative distribution function (CDF) of $A$ given $X$. We will consider interventions under which the observational treatment process $\Pi(\, \cdot \mid X)$ is replaced by some alternative distribution $Q(\, \cdot\mid X)$. To avoid notational burden, we will often simply write $\Pi$ and $Q$, hiding the dependence on $X$. We start by formally defining ``traditional'' stochastic interventions in terms of generalized policies.

\begin{example}[Pure stochastic policy]
    For any distribution $Q$, the \textit{pure stochastic policy} under $Q$, denoted as $d_Q$, assigns treatment according to an independent random draw from $Q$. Concretely, $d_Q(X, V) \coloneqq Q^{-1}(V \mid X)$, where $V \sim \mathrm{Unif}(0,1)$.
\end{example}

In the above example, $Q^{-1}$ denotes the usual (conditional) quantile function associated with $Q$. Although we refer to $d_Q$ as a ``pure stochastic'' policy, it is often ``dynamic'' in that it can depend on $X$. We will often specialize to the exponential tilted target distribution $Q_{\delta} \ll \Pi$ (where $\ll$ denotes absolute continuity between measures), for $\delta \in \mathbb{R}$, defined via
\begin{equation}\label{eq:exp-tilt}
d Q_{\delta}(a \mid X) = \frac{e^{\delta a} \, d\Pi(a \mid X)}{\mathbb{E}_P(e^{\delta A} \mid X)}, \text{ for } a \in \mathcal{A},
\end{equation}
assuming $\mathbb{E}_P(e^{\delta A} \mid X) = \int e^{\delta a} \, d\Pi(a \mid X) < \infty$, i.e., the (conditional) moment generating function of $A$ is finite at $\delta$. This is also called the \textit{Esscher transform} \citep{esscher1932}, and corresponds to incremental interventions when $\mathcal{A} = \{0,1\}$ \citep{kennedy2019}. Applications for continuous treatments have been considered in~\citet{diaz2020} and~\citet{schindl2024}. It can be seen (e.g., see \citet{rakshit2024}) that $Q_{\delta_2}$ stochastically dominates $Q_{\delta_1}$ whenever $\delta_2 > \delta _1$. Since $\Pi \equiv Q_0$, the pure stochastic policy for $Q_{\delta}$ results in treatments whose distribution is pushed upward (downward) compared to the observational treatment process $\Pi$, when $\delta > 0$ ($\delta < 0$). Exponential tilts are smooth and well-behaved, but also are ``closest'' to the observational treatment distribution in a specific sense---in Appendix~\ref{app:distance}, we show that $Q_{\delta}$ minimizes the Kullback-Leibler divergence with respect to $\Pi$, subject to a mean constraint.

\subsection{Motivating example}\label{sec:motivating-example}
As alluded to in the introduction, the behavior of pure stochastic policies can differ substantially from other generalized policies that yield the same target treatment distribution. Often, $d_Q$ will even be suboptimal in terms of partial identification of its effect. To demonstrate these phenomena, consider the following illustrative scenario. Suppose the observed data distribution is described by $(X, A, Y) \in [0,1] \times \{0,1\} \times \mathbb{R}$, where $X \sim \mathrm{Unif}(0,1)$, $\pi(X) \coloneqq P[A = 1 \mid X] = X$, and $\mu_a(X) \coloneqq \mathbb{E}(Y \mid X, A = a) = (2a-1)X$, for $a \in \{0,1\}$. Say we are interested in incremental effects~\citep{kennedy2019}, where $\pi$ is ``replaced'' by $q_{\delta} = \frac{e^{\delta} \pi}{e^{\delta} \pi + (1 - \pi)}$, for some range of values $\delta$; more formally, we want to learn about $\mathbb{E}(Y(d_{q_{\delta}}))$ from the observed data distribution $P$. If we make an unconfoundedness assumption (i.e., $A \ind Y(a) \mid X$), then $\mathbb{E}(Y(d_{q_{\delta}})) = \mathbb{E}((1 - q_{\delta})\mu_0 + q_{\delta}\mu_1) \eqqcolon \tau_{\delta}(P)$. On the other hand, if there is unmeasured confounding, then $\mathbb{E}(Y(d_{q_{\delta}}))$ is not point identified. However, we can bound this quantity under a given sensitivity model. For concreteness, suppose that for some $\Gamma \in [0, \infty)$, we assert
\[\big|\mathbb{E}(Y(a) \mid X, A = 1 - a) - \mu_a(X)\big| \leq \Gamma, \text{ for } a \in \{0,1\}.\]
That is, the value of $\Gamma$ determines the possible extent of violation of the no unmeasured confounding assumption, by bounding the gap between the mean of $Y(a)$ for those treated with $A = a$ and those treated with $A = 1 - a$.

Using the methodology we develop in Section~\ref{sec:methods}, we can construct sharp nonparametric bounds on $\mathbb{E}(Y(d_{q_{\delta}}))$ based on the above sensitivity model. In Figure~\ref{fig:toy-example}, we plot these bounds (in blue) for $\delta \in [0,3]$ and $\Gamma \in \{0.5, 2\}$. A few observations are of note. First, we see that, unsurprisingly, the bounds are narrower when $\Gamma$ is smaller. Indeed, as $\Gamma \to 0$, the bounds collapse to $\tau_{\delta}$ (the black line). Second, the width of the bounds decreases slightly as $\delta \to 0$, but the bounds do not collapse to a point. While $q_0 \equiv \pi$, in general $\mathbb{E}(Y(d_{\pi}))$ will not equal $\tau_0 \equiv \mathbb{E}(Y)$. This stems from the fact that when there is unmeasured confounding, the pure stochastic intervention $d_{\pi}$ deviates from the observational regime in that $d_\pi$ is unconfounded by construction---it depends only on $X$ and auxiliary randomness. Third, we also plot bounds (in red) for a ``maximal'' $q_{\delta}$-policy, $d_{q_{\delta}}^*$, to be described in Section~\ref{sec:methods}, which under unconfoundedness (or $\Gamma = 0$) also satisfies $\mathbb{E}(Y(d_{q_{\delta}}^*)) = \tau_{\delta}$. Bounds for this effect decrease much more drastically in width as $\delta \to 0$, and narrow smoothly to a point at 0. Intuitively, this is accomplished by tying $d_{q_{\delta}}^*$ explicitly to $A$, such that unmeasured confounders are still at play through $A$ itself. In particular, at $\delta = 0$ we have $d_{\pi}^* \equiv A$, so that $\mathbb{E}(Y(d_{\pi}^*)) \equiv \mathbb{E}(Y)$, and point identification is trivially achieved. The construction and interpretation of these alternative generalized policies is discussed in greater detail in the next section.

\begin{figure}[t]
    \centering
    \includegraphics[width=\linewidth]{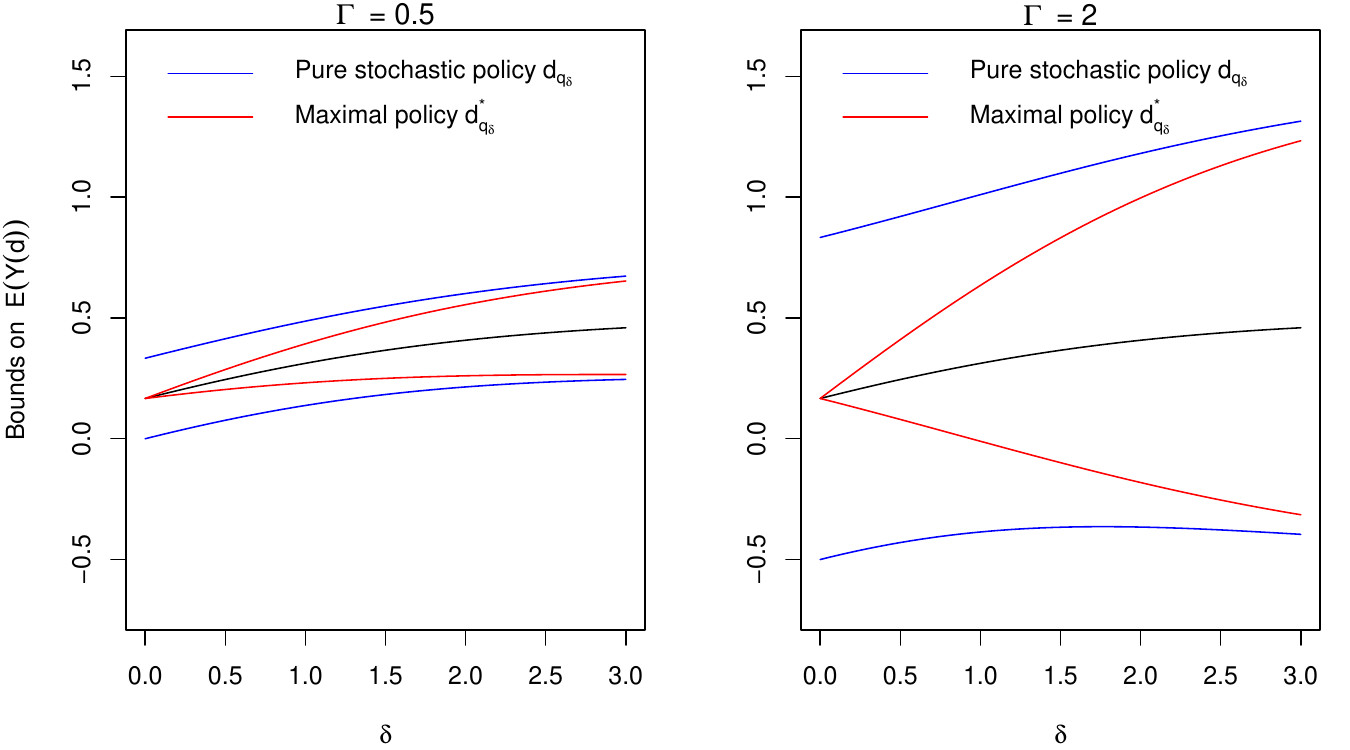}
    \caption{Bounds on $\mathbb{E}(Y(d))$ in motivating example, for incremental interventions. In black is the functional $\mathbb{E}((1 - q_{\delta})\mu_0 + q_{\delta}\mu_1)$, which equals $\mathbb{E}(Y(d_{q_\delta}))= \mathbb{E}(Y(d_{q_\delta}^*))$ under no unmeasured confounding (i.e., $\Gamma = 0$). In blue and red are lower and upper bounds for the effect of the pure stochastic policy, $\mathbb{E}(Y(d_{q_\delta}))$, and the maximal $q_{\delta}$-policy, $\mathbb{E}(Y(d_{q_\delta}^*))$, respectively.}
    \label{fig:toy-example}
\end{figure}

\section{General Methodology}\label{sec:methods}

In this section we lay out our general framework and methodology. In Section~\ref{sec:q-policy}, we define a class of generalized policies whose induced distribution conditional only on covariates is fixed, and provide some examples. In Section~\ref{sec:sens}, we consider sensitivity models that quantify various possible violations of the assumption of no unmeasured confounders, and characterize generalized policy formulations that minimize worst case bounds on mean potential outcomes.
Finally, in Section~\ref{sec:binary}, we explore the binary treatment case in greater detail and provide specialized results.

\subsection{$Q$-policies}\label{sec:q-policy}
For any pair of probability measures on $\mathbb{R}$ (or equivalently, their CDFs), $P_1$ and $P_2$, let $\mathcal{J}(P_1, P_2)$ be the set of joint distributions or \textit{couplings} on $\mathbb{R}^2$ with marginals $P_1$ and $P_2$. To move past pure stochastic policies, we will consider for any target distribution, $Q$, the class of interventions that set treatment to $\widetilde{A}$ where $(A, \widetilde{A}) \mid X$ has distribution $J$ for some $J \in \mathcal{J}_Q(X) \coloneqq \mathcal{J}(\Pi( \, \cdot \mid X), Q( \, \cdot \mid X))$. Equivalently, we can work with generalized policies. To see this, first note that for any $d \in \mathcal{D}$ such that $P[d(X, A, V) \leq a \mid X] = Q(a \mid X)$ for all $a \in \mathbb{R}$, $(A, d(X, A, V))$ has some valid joint distribution given $X$. Conversely, for a given coupling $J \in \mathcal{J}_Q(X)$, we can construct a generalized policy $d_J$ such that $(A, d_J(X, A, V)) \mid X \sim J$ as follows. Let $K_{X, A} : \mathcal{X} \times \mathcal{A} \times \mathcal{B}(\mathbb{R}) \to [0,1]$ be a Markov kernel encoding the regular conditional probability distribution given $X$ and $A$ associated with $J$ (which exists since $\mathbb{R}$ is a Polish space), i.e., $K_{X,A}(x, a, B) = \mathbb{P}_{(A, \widetilde{A}) \mid X = x\sim J}[\widetilde{A} \in B \mid X = x, A = a]$. Then, defining the conditional CDF $M(a' \mid X, A = a) \coloneqq K_{X,A}(X, a, (-\infty, a'])$, for any $a \in \mathcal{A}$, $a' \in \mathbb{R}$, we can define $d_J(X, A, V) = M^{-1}(V \mid X, A)$, for $V \sim \mathrm{Unif}(0,1)$, which yields the desired property. We stress that this argument serves as a technical justification for working primarily with generalized policies, not as a practical construction device---for the applications we will consider, the generalized policies are described as more explicit functions of $X$, $A$, and $V$.

\begin{definition}
    For any distribution $Q$, define the set of $Q$-policies to be \[\mathcal{D}_Q \coloneqq \big\{d \in \mathcal{D} \,\, \big| \,\, P[d(X, A, V) \leq a \mid X] = Q(a \mid X) \text{ for all } a \in \mathbb{R}\big\},\]
    i.e., the set generalized policies inducing treatment with distribution $Q$, given $X$.
\end{definition}

Note that by Assumption~\ref{ass:aux}, $Q$-policies are independent of all potential outcomes, given $X$ and $A$. That is, we are not considering interventions that can depend on anything beyond covariates and the natural value of treatment. This does not rule out, however, the possibility of dependence on something like $\{\mathbb{E}(Y(a) \mid X, A = a')\}_{a, a' \in \mathcal{A}}$, as this is purely a function of $X$---as these conditional expectations are not always identified, though, they may be more challenging to work with. Note also that we have already seen one example of a $Q$-policy: the pure stochastic policy $d_Q$ belongs to $\mathcal{D}_Q$, corresponding to an independent coupling. Below, we consider two useful examples corresponding to non-trivial couplings. 

\begin{example}[Rank-preserving coupling]\label{ex:inv}
    Suppose $\Pi$ is absolutely continuous with respect to the Lebesgue measure on $\mathbb{R}$ (uniformly over $X$), and is supported on an interval. Consider the MTP
    \[d_{Q}^{\mathrm{mon}}(X, A) \coloneqq Q^{-1}(\Pi(A \mid X)\mid X).\]
    By construction, $d_{Q}^{\mathrm{mon}}(X, A) \mid X \sim Q(\, \cdot \mid X)$, as $\Pi(A \mid X) \mid X \sim \mathrm{Unif}(0,1)$ under our continuity assumption. Moreover, this is a monotonic coupling in that $d_{Q}^{\mathrm{mon}}(X, A) \leq A$ if and only if $Q(A \mid X) \geq \Pi(A \mid X)$, i.e., this represents an intervention that pushes the natural value of treatment in the direction $\mathrm{sign}(\Pi(A \mid X) - Q(A \mid X))$. In particular, if $Q$ stochastically dominates $\Pi$, then $d_{Q}^{\mathrm{mon}}(X, A) \geq A$ almost surely, and vice-versa. Note that this transformation, and a variant theoreof, was also considered in~\citet{diaz2013} to describe an intervention that truncated the treatment value while controlling its distribution. More recently, \citet{rakshit2024} employed this transformation in an instrumental variable setting in order to define local causal effects under a monotonicity assumption.
\end{example}

\begin{example}[Maximal coupling]\label{ex:maximal}
    Let $\rho$ be a dominating measure for $\Pi$ and $Q$ (e.g., $\frac{1}{2}(\Pi + Q)$), and define $\pi = \frac{d\Pi}{d\rho}$, $q = \frac{dQ}{d\rho}$. Then, letting $V = (V_1, V_2)$ where $V_1, V_2 \overset{\mathrm{iid}}{\sim} \mathrm{Unif}(0,1)$, consider the generalized policy 
    \[d_Q^*(X, A, V) \coloneqq \mathds{1}(V_1 \leq \widetilde{p}(A \mid X)) A + \mathds{1}(V_1 > \widetilde{p}(A \mid X)) \widetilde{Q}^{-1}(V_2 \mid X),\]
    where $\widetilde{p}(A \mid X) = \min{\left\{1, \frac{q( A \mid X)}{\pi(A \mid X)}\right\}}$, and $\widetilde{Q}$ is the CDF for the distribution with density (with respect to $\rho$) given by
    \[\widetilde{q}(a \mid X) = \frac{q(a \mid X) - \min{\{\pi(a \mid X), q(a \mid X)\}}}{\mathrm{TV}(\Pi(\, \cdot \mid X), Q(\, \cdot \mid X))}.\] Here, $\mathrm{TV}(P_1, P_2) = \sup_{B \in \mathcal{B}(\mathbb{R})}|P_1(B) - P_2(B)|$ represents the total variation distance between any two probability measures $P_1$ and $P_2$ on $(\mathbb{R}, \mathcal{B}(\mathbb{R}))$. This construction encodes a so-called \textit{maximal} coupling between $A$ and $d_{Q}^*$, in that among all $J \in \mathcal{J}_Q(X)$, the distribution of $(A, d_{Q}^*(X, A, V)) \mid X$ maximizes the probability of exact equality. In particular, \[P[A =  d_{Q}^*\mid X] = 1 - \mathrm{TV}(\Pi(\, \cdot \mid X), Q(\, \cdot \mid X)) = \int_{\mathcal{A}} \min{\{\pi(a \mid X), q(a \mid X)\}}\, d\rho(a).\]
    We prove these properties in Appendix~\ref{app:maximal} for completeness.
\end{example}


Beyond the explicit couplings in Examples~\ref{ex:inv} and \ref{ex:maximal}, we can consider examples of a different nature: rather than starting with a target treatment distribution $Q$, one can start with an MTP, determine the treatment distribution (given $X$) induced by that MTP, say $Q$, then consider all its associated $Q$-policies. In a sense, this is the direction pursued in~\citet{young2014} and~\citet{diaz2023}, although these authors only consider the pure stochastic policy associated with a given MTP. The literature on MTPs includes a handful of specific examples, such as threshold interventions, where treatment is moved to exceed a pre-specified threshold if it initially falls below it, or vice-versa~\citep{taubman2009}, and shift interventions, where treatment is moved up (or down) by a fixed amount, unless this would result in leaving the support of the distribution~\citep{diaz2012, haneuse2013}. For both threshold and shift interventions, one can take the resulting treatment distributions and consider alternative generalized policies that would yield the same target distribution, e.g., the rank-preserving and maximal policies in Examples~\ref{ex:inv} and \ref{ex:maximal}. Another interesting MTP is the so-called ``superoptimal regime'' recently proposed in in~\citet{stensrud2024}:
    \[d^{\mathrm{sup}}(X, A) \in \argmax_{a \in \mathcal{A}} \mathbb{E}(Y(a) \mid X, A),\]
    assuming this maximum is attained. This corresponds to an intervention that maximizes mean outcomes given baseline covariates and the natural value of treatment. \citet{stensrud2024} show that $d^{\mathrm{sup}}$ is identical to the conventional optimal regime in the absence of unmeasured confounding, but otherwise is guaranteed to outperform it. Note that $d^{\mathrm{sup}}$ will induce one intervention distribution $Q^{\mathrm{sup}}$ (i.e., $d^{\mathrm{sup}} \in \mathcal{D}_{Q^{\mathrm{sup}}}$) which may or may not be identified from observational data, and hence this deviates from previous examples where $Q$ was either fixed \textit{a priori} or a function only of the observational treatment distribution $\Pi$.

\subsection{Sensitivity analysis to unmeasured confounding}\label{sec:sens}

We now consider partial identification of $\mathbb{E}(Y(d) \mid X)$, as well as $\mathbb{E}(Y(d))$, for generalized policies $d \in \mathcal{D}_{Q}$. To proceed, we need some additional notation. To begin with, we will generally assume that the observed data distribution $P$ belongs to a large nonparametric model, $\mathcal{P}$. Moreover, under Assumption~\ref{ass:consistency}, $P$ is induced by $P^* \in \mathcal{P}^*$ on the complete data $O^* = (X, A, \{Y(a)\}_{a \in \mathcal{A}})$, and we assume that $\mathcal{P}^*$ is a large nonparametric model on the complete data. For any $P \in \mathcal{P}$, we let $\mathcal{M}(P)$ be the subset of distributions in $P^*$ whose marginal distribution over $O = (X, A, Y)$ is $P$. 

As we will focus on mean counterfactual outcomes, the following functions will play a central role:
\[\nu_a(X, a') \coloneqq \mathbb{E}(Y(a) \mid X, A = a'), \text{ for } a, a' \in \mathcal{A},\]
and $\mu_a(X) \coloneqq \mathbb{E}(Y \mid X, A = a) \equiv \nu_a(X, a)$, under Assumption~\ref{ass:consistency}. In words, $\nu_a(X, a')$ represents, among those with covariates $X$ and treated with $A = a'$, mean outcomes under an intervention that sets treatment to $a$. On the other hand, $\mu_a(X)$ is the usual outcome regression function, the mean of the observed outcome among those with covariates $X$ and treated with $A = a$.

Any assumptions on $O$ or $O^*$ correspond to restrictions on $\mathcal{P}^*$ and $\mathcal{P}$. For example, one often assumes that there are no unmeasured confounders: $A \ind Y(a) \mid X$, for all $a \in \mathcal{A}$. Under this assumption (as well as consistency), one could identify mean counterfactuals from the observed data distribution $P$, i.e., $\mathbb{E}(Y(a) \mid X) = \mu_a(X) = \nu_a(X, a')$ for all $a, a' \in \mathcal{A}$. Without such an assumption, however, this and other related causal quantities are not point identified. One may still be willing to make weaker structural assumptions which will restrict the possible range of, say, $\mathbb{E}(Y(d))$.

\begin{model}[Bounded outcomes]\label{mod:sens-bounded}
    Consider the restriction $Y(a) \in [0,1]$, for all $a \in \mathcal{A}$. This is a fairly mild restriction, as if the range of $Y$ is $[\ell,u]$ for $\ell < u$, then one can simply rescale $\frac{Y - \ell}{u - \ell}$. Since this is a fairly weak assumption, we can only say that $|\nu_a(X, a') - \mu_a(X)| \leq 1$, which may yield wide bounds for a given target causal quantity.
\end{model}

\begin{model}[Outcome-based sensitivity model]\label{mod:sens-outcome}
Consider the sensitivity model 
\begin{equation*}
    |\nu_a(X, a') - \mu_a(X)| \leq \Gamma, \text{ for } a, a' \in \mathcal{A}
\end{equation*} for some $\Gamma \geq 0$. Note that under Assumption~\ref{ass:consistency}, $\nu_a(X, a) = \mu_a(X)$ for all $a$, so we need only assert the above inequality when $a' \neq a$. This model directly encodes the possible degree of violation of the unconfoundedness assumption on the scale of mean outcomes. In other words, the parameter $\Gamma$ restricts how different the mean of $Y(a)$ may be, comparing those treated with $a$ and those treated with some other level $a'$, conditioning also on $X$.
    
\end{model}

\begin{model}[Continuous outcome-based model]\label{mod:sens-outcome-cont}
Consider the alternative outcome-based model given by
\begin{equation*}
    |\nu_a(X, a') - \mu_a(X)| \leq \Gamma|a' - a|^p, \text{ for } a, a' \in \mathcal{A}
\end{equation*} for some $p \geq 1$ and $\Gamma \geq 0$. This model imposes additional structure in that the mean of $Y(a)$ for the subsets of subjects treated with $A = a'$ and $A = a$ are similar when $a'$ is close to $a$. This model is perhaps more appropriate in some cases when examining sensitivity to unmeasured confounding with continuous treatment.
    
\end{model}

\begin{model}[Treatment-based sensitivity model]\label{mod:sens-prop}
    Let $H_a(\, \cdot \mid X, Y(a))$ denote the distribution of $A$ given $X$ and $Y(a)$ for each $a \in \mathcal{A}$. Assume $H_a \ll \rho$ for some fixed dominating measure $\rho$, and let $\eta_a(\, \cdot \mid X, Y(a)) = \frac{dH_a(\cdot \mid X, Y(a))}{d\rho}$. Suppose also that $\Pi \ll \rho$, and let $\pi(\, \cdot \mid X) = \frac{d\Pi( \cdot \mid X)}{d\rho}$. Consider now the propensity score odds ratio-based model given by
    \[\frac{1}{\Gamma} \leq \frac{\pi(a' \mid X) / \pi(a \mid X)}{\eta_a(a' \mid X, Y) / \eta_a(a \mid X, Y)} \leq \Gamma, \text{ for all } a, a'\in \mathcal{A}, \text{ w.p.1.},\]
    for some $\Gamma \in [1, \infty)$. This is similar to a density-ratio model proposed in~\citet{bonvini2022b}, and generalizes the odds ratio models for binary treatments in \citet{tan2006}. If $X$ represented all of the confounders, then this model would hold with $\Gamma = 1$. Otherwise, $Y(a)$ may be predictive of $A$ given $X$ and the density ratio $\pi / \eta_a$ may deviate from 1. We show in Appendix~\ref{app:sens} that
    \begin{equation*}
        \nu_a(X, a') = \mathbb{E}(Y g(a', a \mid X, Y) \mid X, A = a), \text{ for all } a, a' \in \mathcal{A},
    \end{equation*}
    where $g(a', a \mid X, Y) \coloneqq \frac{\pi(a \mid X) / \pi(a' \mid X)}{\eta_a(a \mid X, Y) / \eta_a(a' \mid X, Y)}$, using Bayes' rule.
    We also show in Appendix~\ref{app:sens} that $\mathbb{E}(g(a', a \mid X, Y) \mid X, A = a) = 1$, for all $a, a' \in \mathcal{A}$. The latter constraint is important to keep track of if one aims to construct sharp bounds under this model, as~\citet{dorn2023} show in an analogous model for binary treatment. Indeed, we show that the following bounds are sharp in this model: letting $\gamma_a^-(X) = F_{Y \mid X, A = a}^{-1}(\frac{1}{1 + \Gamma})$, $\gamma_a^+(X) = F_{Y \mid X, A = a}^{-1}(\frac{\Gamma}{1 + \Gamma})$,
    \[\mathbb{E}(Y g_a^-(X, Y) \mid X, A = a) \leq \nu_a(X, a') \leq \mathbb{E}(Y g_a^+(X, Y) \mid X, A = a),\]
    where 
    \[g_a^-(X, Y) = \frac{1}{\Gamma} \mathds{1}(Y > \gamma_a^-(X)) + \Gamma \mathds{1}(Y \leq \gamma_a^-(X)), \text{ and }\]
    \[g_a^+(X, Y) = \frac{1}{\Gamma} \mathds{1}(Y < \gamma_a^+(X)) + \Gamma \mathds{1}(Y \geq \gamma_a^+(X)),\]
    assuming the CDF of $Y$ given $X$ and $A = a$ is strictly increasing at the quantiles $\gamma_a^-(X)$ and $\gamma_a^+(X)$. These bounds are new and may be of independent interest, as they extend the results of~\citet{dorn2023} to arbitrary discrete and continuous treatments.
\end{model}

Now, we return to the question of partial identification of $\mathbb{E}(Y(d) \mid X)$ and $\mathbb{E}(Y(d))$, but first assert two conditions on the alternative treatment distribution $Q$.

\begin{assumption}[Absolute continuity] \label{ass:continuity}
    $Q(\, \cdot \mid X) \ll \Pi(\, \cdot \mid X)$, almost surely in $X$.
\end{assumption}

\begin{assumption}[Identified alternative] \label{ass:alt-ident}
    The distribution $Q$ is identified under $P$.
\end{assumption}

In words, Assumption~\ref{ass:continuity} requires that the probability measure on $\mathbb{R}$ defined by the CDF $Q$ is absolutely continuous with respect to that for $\Pi$, such that $Q$ does not have positive mass on a subset for which $\Pi$ assigns no mass. This takes the place of the usual positivity assumption, which---in its strongest form---would require that $\Pi$ has density (with respect to some dominating measure) bounded away from zero on the entire support. Such a restriction is common when aiming to estimate the dose response curve $\mathbb{E}(Y(a))$, for $a \in \mathcal{A}$ \citep{kennedy2017}, but is overly restrictive for interventions based on generalized policies. Assumption~\ref{ass:alt-ident}, on the other hand, is a practical choice, so that the complexity of partial identification is limited solely to the distribution of $\{Y(a)\}_{a \in \mathcal{A}}$ given $(X, A)$, and $Q$ may be treated as fixed.

\begin{proposition}\label{prop:ident}
    Suppose Assumptions~\ref{ass:consistency}-\ref{ass:aux} hold, let $Q$ satisfy Assumption~\ref{ass:continuity}-\ref{ass:alt-ident}, and define the (identified) function $t_Q(X) = \int \mu_a(X) \, dQ(a \mid X)$. Then for any $d \in \mathcal{D}_Q$,
    \begin{align*}
        \mathbb{E}(Y(d) \mid X) 
        &= t_Q(X) + \int \mathds{1}(a' \neq a)\{\nu_{a}(X, a') - \mu_a(X)\}\, dP_{A, d \mid X}(a', a \mid X) \\
        &= t_Q(X) + \mathbb{E}(\mathds{1}(A \neq d)\{\nu_d(X, A) - \mu_d(X)\} \mid X),
    \end{align*}
    where $P_{A, d \mid X} \in \mathcal{J}_Q(X)$ is the joint distribution of $(A, d(X, A, V))$, given $X$.
\end{proposition}

The result of Proposition~\ref{prop:ident} sets the foundation for subsequent results, and several remarks are in order. First, notice that if $X$ represents a sufficient set of measured confounders, then $\mathbb{E}(Y(d) \mid X) = t_Q(X) \equiv \mathbb{E}(\mu_d(X)\mid X)$, for any $d \in \mathcal{D}_Q$, and as a consequence,

\[\mathbb{E}(Y(d)) = \mathbb{E}\left(\int \mu_a(X) \, dQ(a \mid X)\right) \eqqcolon \tau_Q(P).\]

\noindent The functional $\tau_Q$ has been the usual object of interest when studying stochastic interventions. Typically, this has been interpreted as the effect of the pure stochastic intervention $d_Q$ under unconfoundedness, but given Proposition~\ref{prop:ident}, $\tau_Q$ equally represents the effect of \textit{any} generalized policy that induces treatment distribution $Q$. Second, when there is unmeasured confounding, the latter statement fails, in that the mean effect of generalized policies $d \in \mathcal{D}_Q$ may vary. In particular, the causal effect of the generalized policy $d$ depends on an interplay between the degree of unmeasured confounding (i.e., through $\nu_a(X, a') - \mu_a(X)$), and the joint distribution of $A$ and $d(X, A, V)$. Third, as we saw in the motivating example of Section~\ref{sec:motivating-example}, the pure stochastic policy $d_Q$ may suffer more from unmeasured confounding than other generalized policies, in that bounds on $\mathbb{E}(Y(d_Q) \mid X)$ may be wider than those for other $Q$-policies $d \in \mathcal{D}_Q$. As another illustration, consider sensitivity Model~\ref{mod:sens-outcome}, take $Q = \Pi$, and suppose $\Pi$ is continuous so that $P[A \neq d_{\Pi} \mid X] = 1$, by independence. Then $\mathbb{E}(Y(d_{\Pi}) \mid X) - t_{\Pi}(X) \in [-\Gamma, \Gamma]$, whereas the maximal $\Pi$-policy satisfies $\mathbb{E}(Y(d_{\Pi}^*) \mid X) - t_{\Pi}(X)= 0$, since $P[A = d_\Pi^* \mid X] = 1$. Fourth and finally, the problem of finding the optimal $d \in \mathcal{D}_Q$, in terms of minimizing worst-case bounds on $\mathbb{E}(Y(d)\mid X)$, shares connections with \textit{optimal transport} \citep{villani2009}. To see this, suppose a sensitivity model implies bounds $|\nu_a(X, a') - \mu_a(X)| \leq \Gamma(a, a', X)$, for all $a' \neq a$, for some $\Gamma : \mathcal{A}^2 \times \mathcal{X} \to [0,\infty)$. Then
a Monge-Kantorovich formulation arises when minimizing the bound length (i.e., the possible deviation of $\mathbb{E}(Y(d) \mid X)$ from $t_Q(X)$) across all generalized policies in $\mathcal{D}_Q$:

\[\min_{J \in \mathcal{J}_Q(X)} \int \mathds{1}(a' \neq a) \Gamma(a, a', X) \, dJ(a, a' \mid X).\]

\noindent Thus, when characterizing optimal couplings and corresponding $Q$-policies, it will often be possible to leverage results from the optimal transport literature. Before providing such results, we make two more definitions.

\begin{definition}\label{def:sharp}
    Given a complete data distribution model $\mathcal{P}^*$, and a functional $\chi: \mathcal{P}^* \to \mathbb{R}$, the observed data functionals $L, U:\mathcal{P} \to \mathbb{R} \cup \{-\infty, \infty\}$ defined by
    \[L(P) = \inf_{P^* \in \mathcal{M}(P)} \chi(P^*), \, U(P) = \sup_{P^* \in \mathcal{M}(P)} \chi(P^*)\]
    are called \textit{tight} lower and upper bounds on $\chi$ in the model $\mathcal{P}$, respectively, and are called \textit{sharp} at $P$ if each is attained by some $P^* \in \mathcal{M}(P)$.
\end{definition}

\begin{definition}\label{def:optimal}
    For a target distribution $Q$, let $\chi_d: \mathcal{P}^* \to \mathbb{R}$ be a complete data functional that depends on $d \in \mathcal{D}_Q$, and let $L_d, U_d : \mathcal{P} \to \mathbb{R}$ be tight bounds on $\chi_d$ for each $d \in \mathcal{D}_Q$. We say $d_{\chi} \in \mathcal{D}_Q$ is an \textit{optimally tight} $Q$-policy at $P$ if
    \[d_{\chi} \in \argmin_{d \in \mathcal{D}_Q} \, \{U_d(P) - L_d(P)\},\]
    and is \textit{optimally sharp} if $L_{d_\chi}, U_{d_\chi}$ are sharp at $P$.
\end{definition}

While Definition~\ref{def:sharp} gives the usual notion of the best possible bounds in a given sensitivity model, in Definition~\ref{def:optimal} we formalize when a $Q$-policy is ``best'' for a given target functional and sensitivity model. With these two definitions in hand, the following results cover many applications.

\begin{theorem}\label{thm:maximal}
The maximal $Q$-policy $d_Q^*$ in Example~\ref{ex:maximal} satisfies
    \[d_Q^* \in \argmin_{d \in \mathcal{D}_Q} P[A \neq d \mid X], \text{ and } P[A \neq d_Q^* \mid X] = \mathrm{TV}(\Pi(\, \cdot \mid X), Q(\, \cdot \mid X)) \eqqcolon \mathrm{TV}_{\Pi, Q}(X).\]
Supposing $\Gamma^-, \Gamma^+ : \mathcal{X} \to [0, \infty)$ yield bounds $-\Gamma^-(X) \leq \nu_a(X, a') - \mu_a(X) \leq \Gamma^{+}(X)$, for all $a'\neq a$, then under the assumptions of Proposition~\ref{prop:ident},
\[\mathbb{E}(Y(d_Q^*) \mid X) \in [t_Q(X) - \Gamma^-(X)\mathrm{TV}_{\Pi, Q}(X), t_Q(X) + \Gamma^+(X)\mathrm{TV}_{\Pi, Q}(X)].\]
Moreover, these are the narrowest possible bounds of this kind (i.e., based on $\Gamma^-, \Gamma^+$) for any $d \in \mathcal{D}_Q$, so that $d_Q^*$ is optimally tight/sharp if $\Gamma^-$ and $\Gamma^+$ are tight/sharp.
\end{theorem}

The result of Theorem~\ref{thm:maximal} immediately gives sharp bounds on $\mathbb{E}(Y(d_Q^*)\mid X)$ and $\mathbb{E}(Y(d_Q^*))$ for sensitivity Models~\ref{mod:sens-bounded} and~\ref{mod:sens-outcome}. For Model~\ref{mod:sens-outcome}, we have $\mathbb{E}(Y(d_Q^*) \mid X) \in t_Q(X) \pm \Gamma \cdot \mathrm{TV}_{\Pi, Q}(X)$. Moreover, as there are no additional constraints imposed by the sensitivity model with respect to the distribution of $X$, we can simply marginalize these conditional bounds to obtain sharp bounds for the marginal quantity, $\mathbb{E}(Y(d_Q^*)) \in \tau_Q \pm \Gamma \cdot \mathbb{E}(\mathrm{TV}_{\Pi, Q}(X))$. The bounded outcomes assumption in Model~\ref{mod:sens-bounded} is a special case with $\Gamma = 1$ for both the conditional and marginal bounds. As $d_Q^*$ is optimally sharp in these models, these are the best bounds we can hope to achieve for any generalized policies that yield treatment distribution $Q$ given $X$. This last result, however, is not sufficient to handle more complex sensitivity models such as Model~\ref{mod:sens-outcome-cont}. For this case, we have the following result.

\begin{theorem}\label{thm:monotonic}
    Assume the conditions of Example~\ref{ex:inv}, and that $\Pi$ and $Q$ have finite first moments. Then the rank-preserving $Q$-policy $d_Q^{\mathrm{mon}}$ satisfies
    {\small
    \[\min_{d \in \mathcal{D}_Q} \mathbb{E}(h(|A - d|) \mid X) = \mathbb{E}(h(|A - d_Q^{\mathrm{mon}}|) \mid X) = \int_0^1 h\left(|\Pi^{-1}(a \mid X) - Q^{-1}(a \mid X)|\right) \, da \eqqcolon W_{\Pi, Q}^{(h)}(X)\]
    }%
    for any (possibly $X$-dependent) convex function $h: [0,\infty) \to [0,\infty)$ such that $h(0) = 0$. Consequently, supposing for such a function $h$, $|\nu_a(X, a') - \mu_a(X)| \leq h(|a' - a|)$ for all $a, a' \in \mathcal{A}$, then under the assumptions of Proposition~\ref{prop:ident},
    \[\mathbb{E}(Y(d_Q^{\mathrm{mon}}) \mid X) \in [t_Q(X) - W_{\Pi, Q}^{(h)}(X),t_Q(X) + W_{\Pi, Q}^{(h)}(X)].\]
    These are the narrowest possible bounds of this kind (i.e., based on $h$) for any $d \in \mathcal{D}_Q$, so that $d_Q^{\mathrm{mon}}$ is optimally tight/sharp if the bounds $|\nu_a(X, a') - \mu_a(X)| \leq h(|a' - a|)$ are tight/sharp.
\end{theorem}

From Theorem~\ref{thm:monotonic}, we obtain the following sharp bounds under Model~\ref{mod:sens-outcome-cont}:
\[\mathbb{E}(Y(d_{Q}^{\mathrm{mon}}) \mid X) \in t_q(X) \pm \Gamma \cdot W_p^p(\Pi(\, \cdot \mid X), Q(\, \cdot \mid X)),\]
where $W_p(P_1, P_2)$ is the $p$-Wasserstein distance between distributions $P_1$ and $P_2$. Since $d_Q^{\mathrm{mon}}$ is optimally sharp in this model, these bounds are the narrowest that can be achieved among all generalized policies $d \in \mathcal{D}_Q$. 

We end this section by noting that neither Theorem~\ref{thm:maximal} nor~\ref{thm:monotonic} is suitable to yield optimally sharp policies for the propensity score-based sensitivity Model~\ref{mod:sens-prop}, since bounds for this model are of the form $-\Gamma_a^-(X) \leq \nu_a(X, a') - \mu_a(X) \leq \Gamma_a^+(X)$, for all $a' \neq a$, for some non-negative functions $\Gamma_a^-, \Gamma_a^+$. That is, the bounds depend on $a$ unlike in the statement of Theorem~\ref{thm:maximal}, and do not depend on $a'$ as required for Theorem~\ref{thm:monotonic}. Optimal policies for such bounds can be determined from the solution of a linear optimization problem, however whether this solution has a simple closed form is unclear, and we leave this problem for future work. We show in Section~\ref{sec:binary}, though, that when the treatment is binary, the maximal $Q$-policy remains optimal even for these more complicated bounds. Moreover, in the general case, we can still consider bounds on the effect of the maximal $Q$-policy in this sensitivity model; despite lack of optimality, this policy still has favorable properties. To that end, in Appendix~\ref{app:sens}, we develop bounds on $\mathbb{E}(Y(d_Q^*)\mid X)$ for Model~\ref{mod:sens-prop}.


    

\subsection{Binary treatments}\label{sec:binary}
In this subsection, we will consider binary treatments: $\mathcal{A} = \{0,1\}$. In this case, we will write $\pi(X) = P[A = 1 \mid X]$ for the usual propensity score. Any $Q \ll \Pi$  is also supported on $\{0,1\}$, so we can work with alternative distributions $Q$ through their associated intervention propensity scores, $q : \mathcal{X} \to [0,1]$.
Assumption~\ref{ass:continuity} can be restated as follows: $q(X) = \pi(X)$ whenever $\pi(X) \in \{0,1\}$, i.e., when treatment is deterministically 0 or 1 under $\pi$, the alternative distribution $q$ does not shift the treatment. Similarly, we can define the set of $q$-policies as
\[\mathcal{D}_q \coloneqq \big\{d \in \mathcal{D} \,\, \big| \,\, P[d(X, A, V) = 1 \mid X] = q(X) \big\}.\]

We begin by highlighting that the maximal coupling from Example~\ref{ex:maximal} simplifies greatly for binary treatments. We show in Appendix~\ref{app:maximal} that for any $q$, its associated maximal $q$-policy can be equivalently defined as follows:
\[d_q^*(X, A, V) = \begin{cases}
    A & \text{ if } q(X) = \pi(X) \\
    A\mathds{1}\left(V \leq \frac{q(X)}{\pi(X)}\right) & \text{ if } q(X) < \pi(X) \\
    A + (1 - A)\mathds{1}\left(V > \frac{1 - q(X)}{1 - \pi(X)}\right) & \text{ if } q(X) > \pi(X)
\end{cases},\]
where $V \sim \mathrm{Unif}(0,1)$. Note that $d_q^*$ takes a very appealing form: it results from a ``push'' of the natural treatment $A$ in the direction of $\mathrm{sign}{\{q - \pi\}}$, that is, if $q > \pi$, some individuals for whom $A = 0$ are moved to the treatment group under the intervention, and if $q < \pi$, some treated individuals are moved to the control condition. In this respect, $d_q^*$ is rank-preserving (i.e., respects monotonicity), similar to the policy in Example~\ref{ex:inv}. For example, consider the incremental propensity score distribution for $\delta \in \mathbb{R}$, given by
\begin{equation}\label{eq:incremental}
    q_{\delta}(X) = \frac{e^{\delta}\pi(X)}{e^{\delta}\pi(X) + (1 - \pi(X))}.
\end{equation}
By construction, this satisfies Assumption~\ref{ass:continuity} and $\frac{q_\delta}{1 - q_{\delta}} \big/ \frac{\pi}{1 - \pi} \equiv e^{\delta}$, i.e., the odds of treatment are multiplied by $e^\delta$ under the intervention. Thus, when $\delta > 0$, since $q_{\delta} \geq \pi$ everywhere, we have $d_{q_{\delta}}^* \geq A$. Similarly, when $\delta < 0$, $d_{q_{\delta}}^* \leq A$. 

Even beyond considerations related to unmeasured confounding, one may prefer $d_{q_{\delta}}^*$ over the pure stochastic policy $d_{q_{\delta}}$ proposed in~\citet{kennedy2019} simply due to its nice interpretation. In general, maximal $q$-policies are arguably more predictable and parsimonious, in that treatment is changed for fewer subjects: 
\[P[A \neq d_q^* \mid X] = |q - \pi| \leq \pi (1 - q) + (1 - \pi)q = P[A \neq d_q \mid X].\]
We will now see that maximal $q$-policies are also favorable from the perspective of partial identification. Let $\mu_a(X) = \mathbb{E}(Y \mid X, A = a)$ and $\nu_a(X) = \mathbb{E}(Y(a) \mid X, A = 1 - a)$, for $a \in \{0,1\}$, similar to our notation from Section~\ref{sec:sens}. We begin with by specializing Proposition~\ref{prop:ident} to binary treatments.

\begin{corollary}\label{cor:ident-bin}
    Suppose Assumptions~\ref{ass:consistency}-\ref{ass:aux} hold, and let $q$ satisfy Assumptions~\ref{ass:continuity}-\ref{ass:alt-ident}. Define $t_q(X) = \mu_0(X)(1 - q(X)) + \mu_1(X) q(X)$. For any $d \in \mathcal{D}_q$, we have
    \[\mathbb{E}(Y(d) \mid X) = t_q + \pi(1 - \widetilde{q}_d)(\nu_0 - \mu_0) + (q - \pi \widetilde{q}_d)(\nu_1 - \mu_1),\]
    where $\widetilde{q}_d(X) \coloneqq P[d(X, A, V) = 1 \mid X, A = 1]$.
\end{corollary}

\begin{theorem}\label{thm:maximal-bin}
    Suppose $\Gamma_a^-, \Gamma_a^+ : \mathcal{X} \to [0, \infty)$, for $a \in \{0,1\}$ yield bounds
    \begin{equation}\label{eq:ineq-maximal-bin}
    -\Gamma_a^-(X) \leq \nu_a(X) - \mu_a(X) \leq \Gamma_a^+(X), \text{ for } a \in \{0,1\}.
    \end{equation}
    Then under the assumptions of Corollary~\ref{cor:ident-bin},
    \[\mathbb{E}(Y(d_q^*) \mid X) \in [t_q - \Gamma_s^- |q - \pi|, t_q + \Gamma_s^+ |q - \pi|],\]
    where $s(X) \coloneqq \mathds{1}(q(X) > \pi(X))$. Moreover, these are the narrowest possible bounds of this kind (i.e., based on $\Gamma_a^-, \Gamma_a^+$) for any $d \in \mathcal{D}_q$, and $d_q^*$ is optimally tight/sharp if the inequalities in~\eqref{eq:ineq-maximal-bin} are tight/sharp.
\end{theorem}

The bounds on $\mathbb{E}(Y(d_q^*) \mid X)$ expressed in Theorem~\ref{thm:maximal-bin} have length $(\Gamma_s^- + \Gamma_s^+)|q - \pi|$. Contrast this with corresponding bounds for the pure stochastic policy $d_q \coloneqq \mathds{1}(V \leq q(X))$ for $V \sim \mathrm{Unif}(0,1)$: it follows easily from Corollary~\ref{cor:ident-bin} that
\[\mathbb{E}(Y(d_q) \mid X) - t_q \in [-\Gamma_0^- \pi(1 - q) - \Gamma_1^- (1 - \pi)q, \Gamma_0^+ \pi(1 - q) + \Gamma_1^+ (1 - \pi)q],\]
which are, in general, wider---equality between the two approaches holds if and only if $\pi$ or $q$ is degenerate at $X$, i.e.,  $\pi(X) \in \{0,1\}$ or $q(X) \in \{0,1\}$. Note also that if $\Gamma_a^-, \Gamma_a^+$ for $a \in \{0,1\}$ give sharp bounds~\eqref{eq:ineq-maximal-bin}, for all $x \in \mathcal{X}$, and there are no additional constraints, Theorem~\ref{thm:maximal-bin} also implies sharp bounds on the marginal counterfactual mean quantity:
\[\mathbb{E}(Y(d_q^*)) \in [\tau_q - \mathbb{E}(\Gamma_s^- |q - \pi|), \tau_q + \mathbb{E}(\Gamma_s^+ |q - \pi|)], \]
where $\tau_q = \mathbb{E}(t_q(X))$.

These results suffice for a large number of sensitivity models, including Models~\ref{mod:sens-bounded}--\ref{mod:sens-prop} specialized to binary treatments. Model~\ref{mod:sens-outcome} can be equivalently written as $|\nu_a(X) - \mu_a(X)| \leq \Gamma$, and Model~\ref{mod:sens-bounded} implies this bound with $\Gamma = 1$. Model~\ref{mod:sens-outcome-cont}, taken literally, is just a restatement of Model~\ref{mod:sens-outcome}, since $\mathcal{A}= \{0,1\}$. Indeed, these two models only differ for non-binary treatments. Finally, the odds ratio-based Model~\ref{mod:sens-prop} can be written
\[\frac{1}{\Gamma} \leq \frac{\pi(X) / (1 - \pi(X))}{\eta_a(X, Y) / (1 - \eta_a(X, Y))} \leq \Gamma,\]
where $\eta_a(X, y) = P[A = 1 \mid X, Y(a) = y]$ for $a \in \{0,1\}$. This has been referred to as the ``marginal sensitivity model'' by~\citet{zhao2019}, though a version of it was originally proposed by~\citet{tan2006}. As we saw, this implies sharp bounds
\[\mathbb{E}(Y\{g_a^-(X,Y) - 1\} \mid X, A = a) \leq \nu_a(X) - \mu_a(X) \leq \mathbb{E}(Y\{g_a^+(X,Y) - 1\} \mid X, A = a),\]
in the notation of Model~\ref{mod:sens-prop}. Thus, in all cases the bounds take the required form for application of Theorem~\ref{thm:maximal-bin}, and the maximal $q$-policy is optimally sharp.

\section{Estimation and Inference}\label{sec:estimation}

In this section, we develop estimators of the lower and upper bounds on marginal effects $\mathbb{E}(Y(d))$ for optimal $Q$-policies $d$, in accordance with our results in Section~\ref{sec:methods}. We focus on sensitivity Models~\ref{mod:sens-bounded}--\ref{mod:sens-prop}, and exponentially tilted alternative treatment distributions,  $Q_{\delta}$, as defined in~\eqref{eq:exp-tilt} and~\eqref{eq:incremental}. In Section~\ref{sec:est-cont}, we consider bound estimators for continuous treatment $A$, and in Section~\ref{sec:est-bin}, we construct specialized estimators for the case that treatment is binary.

As a matter of notation, we will write $\mathbb{P}_n(f) = \frac{1}{n}\sum_{i = 1}^n f(O_i)$, $\lVert f \rVert_{\infty} = \mathrm{ess\,sup}(f)$, and $\lVert f \rVert^2 = \int f(o)^2 \, dP(o)$ for the empirical mean, $L_{\infty}(P)$ norm, and squared $L_2(P)$ norm of $f$, respectively. Throughout this section, we will assume that the nuisance functions (treatment distribution and outcome models) are fit in a separate sample, independent of $(O_1, \ldots, O_n)$. In the usual case where one has access only to a single sample, one can perform sample splitting and cross-fitting~\citep{bickel1988, robins2008, zheng2010,chernozhukov2018}. We will assume a single split for simplicity, as results extend in the usual way to estimators averaged across multiple splits.

\subsection{Estimators for continuous treatment}\label{sec:est-cont}
Given our discussion in Section~\ref{sec:sens} following Theorems~\ref{thm:maximal} and \ref{thm:monotonic}, the maximal $Q$-policy is optimally sharp in Model~\ref{mod:sens-outcome}, with bounds on $\mathbb{E}(Y(d_Q^*))$ given by
$\tau_Q \pm \Gamma \cdot \mathbb{E}(\mathrm{TV}_{\Pi, Q}(X))$,
where $\tau_Q = \mathbb{E}(\int \mu_a(X) \, dQ(a\mid X))$, while the rank-preserving $Q$-policy is optimal in Model~\ref{mod:sens-outcome-cont} with bounds on $\mathbb{E}(Y(d_Q^{\mathrm{mon}}))$ given $\tau_Q \pm \Gamma \cdot \mathbb{E}(W_p^p\{\Pi(\, \cdot \mid X), Q(\, \cdot \mid X)\})$. These bounds are sharp for any distribution $Q$, but we will provide detailed estimation results for the exponentially tilted distribution, $Q_{\delta}$, defined in~\eqref{eq:exp-tilt}. We assume throughout that the conditional moment generating function of $A$ exists at $\delta$, i.e., $\mathbb{E}(e^{\delta A} \mid X) < \infty$. This will hold when $A$ is bounded, for example, which will we assume for our estimation results.

The smooth functional $\tau_{\delta} \coloneqq \tau_{Q_{\delta}}$ for continuous treatments has been considered in detail in~\citet{diaz2020} and \citet{schindl2024}, corresponding to the effect of the pure stochastic policy $d_{Q_{\delta}}$ under an assumption of no unmeasured confounding. These works develop efficiency theory and a corresponding nonparametric efficient estimator based on the influence function of $\tau_{\delta}$. We begin by briefly recapping these results here.

\begin{proposition}\label{prop:von-Mises}
    The functional $\tau_{\delta}$ satisfies the von Mises expansion
    \[\tau_{\delta}(\overline{P}) - \tau_{\delta}(P) = -\mathbb{E}_P(\dot{\tau}_{\delta}(O; \overline{P})) + R_{\tau_{\delta}}(\overline{P}, P),\]
    where the influence function is given by
    \[\dot{\tau}_{\delta}(O; P) = \frac{e^{\delta A}}{\alpha_{\delta}(X)}(Y - \gamma_{\delta}(X)) + \gamma_{\delta}(X) - \tau_{\delta},\]
    and the remainder term is
    \[R_{\tau_{\delta}}(\overline{P}, P) = \mathbb{E}_P\left(\{\overline{\gamma}_{\delta} - \gamma_{\delta}\}\left\{1 - \frac{\alpha_{\delta}}{\overline{\alpha}_{\delta}}\right\}\right),\]
    where $\alpha_{\delta}(X) = \mathbb{E}_P(e^{\delta A} \mid X)$ and $\gamma_{\delta}(X) = \frac{\mathbb{E}_P(e^{\delta A}Y \mid X)}{\alpha_{\delta}(X)}$.
\end{proposition}

The result of Proposition~\ref{prop:von-Mises} does not depend on whether treatment or outcomes are discrete or continuous, and importantly the von Mises expansion motivates a simple one-step estimator: $\widehat{\tau}_{\delta} = \mathbb{P}_n[\widehat{\Psi}_{\tau_{\delta}}]$, where
$\widehat{\Psi}_{\tau_{\delta}} = \frac{e^{\delta A}}{\widehat{\alpha}_{\delta}(X)}(Y - \widehat{\gamma}_{\delta}(X)) + \widehat{\gamma}_{\delta}(X)$.
This is the estimator used in~\citet{rakshit2024}, and we note that $\alpha_{\delta}$ and $\gamma_{\delta}$ are composed of regression functions, and thus can be modeled flexibly using one's preferred regression approach. It should be mentioned, though, that~\citet{schindl2024} proposed a one-step estimator specifically for continuous treatment based on a different parametrization---namely, requiring direct estimation of the density of $A$ given $X$---and remarked that it can be more stable in practice.

To estimate bounds in Model~\ref{mod:sens-outcome}, we must put forth an estimator of $\mathbb{E}(\mathrm{TV}_{\Pi, Q_{\delta}}(X))$. By construction, we have that $\frac{d Q_{\delta}}{d \Pi}(A \mid X) = \frac{e^{\delta} A}{\alpha_{\delta}(X)}$, so that
\[\mathbb{E}(\mathrm{TV}_{\Pi, Q_{\delta}}(X)) = 1 - \mathbb{E}\left(\min{\left\{1, \frac{e^{\delta} A}{\alpha_{\delta}(X)}\right\}}\right) = \mathbb{E}\left(\mathds{1}(e^{\delta A}<\alpha_{\delta}(X))\left\{1 - \frac{e^{\delta A}}{\alpha_{\delta}(X)}\right\}\right) \eqqcolon \xi_{\delta}.\]
Unlike $\tau_{\delta}$, the functional $\xi_{\delta}$ is non-smooth due to the presence of the indicator $\mathds{1}(e^{\delta A} < \alpha_{\delta}(X))$. In other words, it does not have an influence function to facilitate efficient nonparametric estimation. Fortunately, if $e^{\delta A}$ has a continuous distribution given $X$, with density bounded in a neighborhood of its mean, then efficient inference is still possible. We will proceed with the following somewhat weaker assumption.

\begin{assumption}[Margin condition 1] \label{ass:margin}
    For some constants $b > 0, C> 0$,
    \[P[|e^{\delta A} - \alpha_{\delta}(X)| \leq t] \leq C \cdot t^b,\]
    for all $t \geq 0$.
\end{assumption}

If $e^{\delta A}$ indeed has continuous distribution given $X$ with density bounded near $\alpha_{\delta}(X)$, then Assumption~\ref{ass:margin} will hold with $b = 1$. This assumption allows us to use a plug-in estimator for the indicator function, i.e., $\mathds{1}(e^{\delta A} < \widehat{\alpha}_{\delta}(X))$. Concretely, we propose $\widehat{\xi}_{\delta} = \mathbb{P}_n[\widehat{\Psi}_{\xi_{\delta}}]$, where
\[\widehat{\Psi}_{\xi_{\delta}} = \left\{\mathds{1}(e^{\delta A} < \widehat{\alpha}_{\delta}(X)) - \widehat{\kappa}_{\delta}(X)\right\}\left\{1 - \frac{e^{\delta A}}{\widehat{\alpha}_{\delta}(X)}\right\},\]
where $\widehat{\kappa}_{\delta}(X)$ is an estimator of $\kappa_{\delta}(X) = \mathbb{E}(e^{\delta A} \mid X, e^{\delta A} < \alpha_{\delta}(X))$. The following result summarizes the asymptotic properties of the proposed estimators $\widehat{\tau}_{\delta}$ and $\widehat{\xi}_{\delta}$.

\begin{theorem}\label{thm:est-cont-inc}
    Assume that $\lVert \widehat{\alpha}_{\delta} - \alpha_{\delta}\rVert + \lVert \widehat{\gamma}_{\delta} - \gamma_{\delta}\rVert + \lVert \widehat{\kappa}_{\delta} - \kappa_{\delta}\rVert = o_P(1)$, and that $P[|Y| \leq C] = P[|A| \leq C] = P[|\widehat{\gamma}_{\delta}| \leq C] = P[|\alpha_{\delta}| \geq 1 / C] = P[|\widehat{\alpha}_{\delta}| \geq 1 / C] =1$, for some $C < \infty$. Then, letting
    $R_{1,n} = \lVert \widehat{\alpha}_{\delta} - \alpha_{\delta}\rVert \cdot \lVert \widehat{\gamma}_{\delta} - \gamma_{\delta}\rVert$,
    we have
    \[\widehat{\tau}_{\delta} - \tau_{\delta} = \mathbb{P}_n\left[\dot{\tau}_{\delta}(O;P)\right] + O_P(R_{1,n}) + o_P(n^{-1/2}).\]
    If Assumption~\ref{ass:margin} holds, then, letting $R_{2,n} = \lVert \widehat{\alpha}_{\delta} - \alpha_{\delta}\rVert \cdot \lVert \widehat{\kappa}_{\delta} - \kappa_{\delta}\rVert$, $R_{3,n} = \lVert \widehat{\alpha}_{\delta} - \alpha_{\delta}\rVert_{\infty}^{1 + b}$,
    \[\widehat{\xi}_{\delta} - \xi_{\delta} = (\mathbb{P}_n - \mathbb{E}_P)\Psi_{\xi_{\delta}} + O_P(R_{2,n} + R_{3,n} ) + o_P(n^{-1/2}),\]
    where $\Psi_{\xi_{\delta}} = \left\{\mathds{1}(e^{\delta A} < \alpha_{\delta}(X)) - \kappa_{\delta}(X)\right\}\left\{1 - \frac{e^{\delta A}}{\alpha_{\delta}(X)}\right\}$. In particular, if $R_{1,n} = o_P(n^{-1/2})$,
    \[\sqrt{n}(\widehat{\tau}_{\delta} - \tau_{\delta}) \overset{d}{\to} \mathcal{N}(0, \mathrm{Var}\{\dot{\tau}_{\delta}(O;P)\}),\]
    and if $R_{2,n} + R_{3,n} = o_P(n^{-1/2})$,
    $\sqrt{n}(\widehat{\xi}_{\delta} - \xi_{\delta}) \overset{d}{\to} \mathcal{N}(0, \mathrm{Var}\{\Psi_{\xi_{\delta}}(O;P)\})$.
\end{theorem}

The above result confirms the double robustness property of the estimator $\widehat{\tau}_{\delta}$ observed in previous work, in that inference depends on the product of errors for $\alpha_{\delta}$ and $\gamma_{\delta}$. The estimator of the expected total variation distance functional, $\widehat{\xi}_{\delta}$, on the other hand, has an additional asymptotic bias term $R_{3,n}$ related to the margin condition---similar bias terms have appeared in other non-smooth functional estimation problems under margin conditions (e.g., \citet{luedtke2016a}). 

Combining these results, the proposed estimator of the bounds on $\mathbb{E}(Y(d_{Q_{\delta}}^*))$ in sensitivity Model~\ref{mod:sens-outcome} is $\widehat{\tau}_{\delta} \pm \Gamma \cdot \widehat{\xi}_{\delta}$. Another result of Theorem~\ref{thm:est-cont-inc} is that when nuisance estimates converge fast enough, valid inference is straightforward: when $R_{1,n} + R_{2,n} = o_P(n^{-1/2})$, an asymptotically valid Wald-based interval for the upper bound is given by

\[\left\{\widehat{\tau}_{\delta} + \Gamma \cdot \widehat{\xi}_{\delta}\right\} \pm z_{1 - \alpha/2}\sqrt{\frac{1}{n} \mathbb{P}_n\left[\left(\widehat{\Psi}_{\tau_{\delta}} - \widehat{\tau}_{\delta} + \Gamma \{\widehat{\Psi}_{\xi_{\delta}} - \widehat{\xi}_{\delta}\}\right)^2\right]}\]
where $z_{1 - \alpha/2}$ is the $(1 - \alpha/2)$-quantile of the standard normal distribution. The confidence interval for the lower bound is defined in the same way, only with $\Gamma$ replaced by $-\Gamma$.

Lastly, we consider estimation of $\mathbb{E}(W_p^p\{\Pi(\, \cdot \mid X), Q_{\delta}(\, \cdot \mid X)\})$, the expected Wasserstein distance functional in the sharp bounds for $\mathbb{E}(Y(d_{Q_{\delta}}^{\mathrm{mon}}))$ in sensitivity Model~\ref{mod:sens-outcome-cont}. For simplicity, we will proceed based on the assumption that $p = 1$, though generalizing to arbitrary $p$ is possible. For this choice, it is well known that
\begin{align*}
    W_1\{\Pi(\, \cdot \mid X), Q_{\delta}(\, \cdot \mid X)\} &= \int_0^1 |\Pi^{-1}(a \mid X) - Q_{\delta}^{-1}(a \mid X)| \, da \\
    &= \int_{\mathbb{R}} |\Pi(a \mid X) - Q_{\delta}(a \mid X)| \, da.
\end{align*}
For $\delta > 0$, $Q_{\delta}$ stochastically dominates $\Pi$, so that when $A \geq 0$,
\begin{align*}
    W_1\{\Pi(\, \cdot \mid X), Q_{\delta}(\, \cdot \mid X)\}
    &= \int_0^{\infty} \{\Pi(a \mid X) - Q_{\delta}(a \mid X)\} \, da \\
    &= \frac{\mathbb{E}(A e^{\delta A} \mid X)}{\mathbb{E}(e^{\delta A} \mid X)} - \mathbb{E}(A \mid X).
\end{align*}
Similarly, when $\delta < 0$, $W_1\{\Pi(\, \cdot \mid X), Q_{\delta}(\, \cdot \mid X)\} = \mathbb{E}(A \mid X)-\frac{\mathbb{E}(A e^{\delta A} \mid X)}{\mathbb{E}(e^{\delta A} \mid X)}$. Thus, bounds on $\mathbb{E}(Y(d_{Q_{\delta}}^{\mathrm{mon}}))$ in Model~\ref{mod:sens-outcome-cont} when $p = 1$ (and $A \geq 0$) are given by $\tau_{\delta} \pm \Gamma \cdot \theta_{\delta}$ when $\delta > 0$, and $\tau_{\delta} \mp \Gamma \cdot \theta_{\delta}$ when $\delta < 0$, where $\theta_{\delta}\coloneqq \mathbb{E}\left(\frac{Ae^{\delta A}}{\alpha_{\delta}(X)} - A\right)$. Letting $\beta(X) = \frac{\mathbb{E}(A e^{\delta A} \mid X)}{\mathbb{E}(e^{\delta A} \mid X)}$,  we take $\widehat{\theta}_{\delta} = \mathbb{P}_n[\widehat{\Psi}_{\theta_{\delta}}(O)]$, where
\[\Psi_{\theta_{\delta}}(O) = \frac{e^{\delta A}}{\alpha_{\delta}(X)}(A - \beta_{\delta}(X)) + \beta_{\delta}(X) - A,\]
and $\widehat{\Psi}_{\theta_{\delta}}$ is obtained by plugging in $(\widehat{\alpha}_{\delta}, \widehat{\beta}_{\delta})$ into this expression.

\begin{theorem}\label{thm:est-cont-inc-wass}
    Assume that $\lVert \widehat{\alpha}_{\delta} - \alpha_{\delta}\rVert = o_P(1)$, and that $P[|A| \leq C] = P[|\widehat{\beta}_{\delta}| \leq C] = P[|\widehat{\alpha}_{\delta}| \geq 1 / C] =1$, for some $C < \infty$. Then
    \[\widehat{\theta}_{\delta} - \theta_{\delta} = (\mathbb{P}_n - \mathbb{E}_P)[\Psi_{\theta_{\delta}}(O)] + O_P(\lVert \widehat{\alpha}_{\delta} - \alpha_{\delta}\rVert \cdot \lVert \widehat{\beta}_{\delta} - \beta_{\delta}\rVert) + o_P(n^{-1/2}),\]
    so that when $\lVert \widehat{\alpha}_{\delta} - \alpha_{\delta}\rVert \cdot \lVert \widehat{\beta}_{\delta} - \beta_{\delta}\rVert = o_P(n^{-1/2})$, $\sqrt{n}(\widehat{\theta}_{\delta} - \theta_{\delta}) \overset{P}{\to} \mathcal{N}(0,\mathrm{Var}\{\Psi_{\theta_{\delta}}(O)\})$.
\end{theorem}

On the basis of this result, for $\delta > 0$ we estimate the bounds on $\mathbb{E}(Y(d_{Q_{\delta}}^{\mathrm{mon}}))$ in Model~\ref{mod:sens-outcome-cont} (for $p = 1$ and non-negative treatment) as $\widehat{\tau}_{\delta} \pm \Gamma \cdot \widehat{\theta}_{\delta}$, with $(1 - \alpha)$-level confidence interval for the upper bound given by
\[\left\{\widehat{\tau}_{\delta} + \Gamma \cdot \widehat{\theta}_{\delta}\right\} \pm z_{1 - \alpha/2} \sqrt{\frac{1}{n}\mathbb{P}_n\left[\left(\widehat{\Psi}_{\tau_{\delta}} - \widehat{\tau}_{\delta} + \Gamma \{\widehat{\Psi}_{\theta_{\delta}} - \widehat{\theta}_{\delta}\}\right)^2\right]},\]
and similarly for the lower bound (as well as cases where $\delta < 0$) with appropriate sign changes.

\subsection{Estimators for binary treatment}\label{sec:est-bin}

In view of Theorem~\ref{thm:maximal-bin}, the maximal $q$-policy is optimally sharp in the sensitivity models we consider. Further, the bounds on $\mathbb{E}(Y(d_q^*))$ for sensitivity Model~\ref{mod:sens-outcome} are
$\tau_q \pm \Gamma\mathbb{E}(|q - \pi|)$, where $\tau_q = \mathbb{E}((1 - q)\mu_0 + q\mu_1)$, and the bounds in Model~\ref{mod:sens-prop} are
\begin{align*}
    & \big[\tau_q + \mathbb{E}\left(\mathbb{E}(Y\{g_s^-(X, Y) - 1\} \mid X, A = s) |q - \pi|\right), \\
    & \quad \tau_q + \mathbb{E}\left(\mathbb{E}(Y\{g_s^+(X, Y) - 1\} \mid X, A = s) |q - \pi|\right)\big]
\end{align*}
where $s(X) = \mathds{1}(q(X) > \pi(X))$. These results hold for any alternative distribution $q$, but we will focus our estimation results on the incremental policy $q_{\delta}$ defined in~\eqref{eq:incremental}.

Recall that \citet{kennedy2019} developed efficiency theory and a one-step estimator for the smooth functional $\tau_{\delta} = \mathbb{E}((1 - q_{\delta}) \mu_0 + q_{\delta}\mu_1)$. Moreover, $q_{\delta}$ satisfies
\[|q_{\delta} - \pi| = \pi\left|\frac{e^{\delta}}{e^{\delta}\pi + (1 - \pi)} - 1\right| = \frac{|e^{\delta} -1|\pi (1 - \pi)}{e^{\delta}\pi + (1 - \pi)}.\]
Thus, estimation of bounds for the maximal $q_{\delta}$-policy effect, $\mathbb{E}(Y(d_{q_{\delta}}^*))$, in the outcome-based sensitivity Model~\ref{mod:sens-outcome} boils down to estimation of the functional $\chi_{\delta} \coloneqq \mathbb{E}\left(\frac{\pi(1 - \pi)}{e^{\delta} \pi + (1 - \pi)}\right)$: by Theorem~\ref{thm:maximal-bin}, bounds are given by $\mathbb{E}(Y(d_{q_{\delta}}^*)) \in \tau_{\delta} \pm \Gamma |e^{\delta} - 1|\chi_{\delta}$.
In the next result, we give the von Mises expansions for $\tau_{\delta}$ and $\chi_{\delta}$. To simplify notation, define
\[\phi_0(O) \coloneqq \frac{1 - A}{1 - \pi(X)}(Y - \mu_0(X)) + \mu_0(X), \text{ and } \phi_1(O) \coloneqq \frac{A}{\pi(X)}(Y - \mu_1(X)) + \mu_1(X).\]
That is, $\phi_a$ is the uncentered influence function of the functional $\mathbb{E}(\mu_a(X))$.

\begin{proposition}\label{prop:von-Mises-bin}
    The functionals $\tau_{\delta}$ and $\chi_{\delta}$ satisfy the following von Mises expansions:
    \[\tau_{\delta}(\overline{P}) - \tau_{\delta}(P) = - \mathbb{E}_P(\dot{\tau}_{\delta}(O; \overline{P})) + R_{\tau_{\delta}}(\overline{P}, P),\]
    \[\chi_{\delta}(\overline{P}) - \chi_{\delta}(P) = - \mathbb{E}_P(\dot{\chi}_{\delta}(O; \overline{P})) + R_{\chi_{\delta}}(\overline{P}, P),\]
    where the influence functions $\dot{\tau}_{\delta}(O;P)$ and $\dot{\tau}_{\delta}(O;P)$ are
    \[\dot{\tau}_{\delta}(O; P) = \frac{e^{\delta} \pi \phi_1 + (1 - \pi)\phi_0}{e^{\delta}\pi + (1 - \pi)} +  \frac{e^{\delta}(\mu_1 - \mu_0)(A - \pi)}{\{e^{\delta}\pi + (1 - \pi)\}^2} - \tau_{\delta}(P),\]
    \[\dot{\chi}_{\delta}(O; P) = \frac{\pi(1 - \pi)}{e^{\delta} \pi + (1 - \pi)} + \frac{\{(1 - \pi)^2 - e^{\delta}\pi^2\}(A - \pi)}{\{e^{\delta}\pi + (1 - \pi)\}^2} - \chi_{\delta}(P),\]
    and the remainder terms are given by
    \[R_{\tau_{\delta}}(\overline{P}, P) = \mathbb{E}_P\left(\{\overline{\gamma}_{\delta} - \gamma_{\delta}\}\left\{1 - \frac{\alpha_{\delta}}{\overline{\alpha}_{\delta}}\right\}\right),\]
    \[R_{\chi_{\delta}}(\overline{P}, P) = \mathbb{E}_P\left(\frac{e^{\delta}}{\overline{\alpha}_{\delta}^2 \alpha_{\delta}}\{\overline{\pi} - \pi\}^2\right),\]
    where $\alpha_{\delta} = e^{\delta}\pi + (1 - \pi)$ and $\gamma_{\delta} = \frac{e^{\delta}\pi\mu_1 + (1 - \pi)\mu_0}{\alpha_{\delta}}$.
\end{proposition}

From Proposition~\ref{prop:von-Mises-bin}, we can immediately obtain one-step estimators for $\tau_{\delta}$ and $\chi_{\delta}$. Estimating $(\widehat{\pi}, \widehat{\mu}_0, \widehat{\mu}_1)$ on separate independent data, and letting $\widehat{\phi}_0 = \frac{1 - A}{1 - \widehat{\pi}}(Y - \widehat{\mu}_0) + \widehat{\mu}_0$, $\widehat{\phi}_1 = \frac{A}{\widehat{\pi}}(Y - \widehat{\mu}_1) + \widehat{\mu}_1$, we take $\widehat{\tau}_{\delta} = \mathbb{P}_n[\widehat{\psi}_{\tau_{\delta}}]$ and $\widehat{\chi}_{\delta} = \mathbb{P}_n[\widehat{\psi}_{\chi_{\delta}}]$, where
\[\widehat{\psi}_{\tau_{\delta}} = \frac{e^{\delta} \widehat{\pi} \widehat{\phi}_1 + (1 - \widehat{\pi})\widehat{\phi}_0}{e^{\delta}\widehat{\pi} + (1 - \widehat{\pi})} +  \frac{e^{\delta}(\widehat{\mu}_1 - \widehat{\mu}_0)(A - \widehat{\pi})}{\{e^{\delta}\widehat{\pi} + (1 - \widehat{\pi})\}^2},\]
\[\widehat{\psi}_{\chi_{\delta}} = \frac{\widehat{\pi}(1 - \widehat{\pi})}{e^{\delta} \widehat{\pi} + (1 - \widehat{\pi})} + \frac{\{(1 - \widehat{\pi})^2 - e^{\delta}\widehat{\pi}^2\}(A - \widehat{\pi})}{\{e^{\delta}\widehat{\pi} + (1 - \widehat{\pi})\}^2}.\]
Note that $\widehat{\tau}_{\delta}$ is identical to the estimator proposed in~\citet{kennedy2019}. 

\begin{theorem}\label{thm:est-bin-inc}
    Assume that $\lVert \widehat{\pi} - \pi\rVert + \lVert \widehat{\mu}_0 - \mu_0\rVert + \lVert \widehat{\mu}_1 - \mu_1\rVert = o_P(1)$, and that $P[|Y| \leq C] = P[|\widehat{\mu}_0| \leq C] = P[|\widehat{\mu}_1| \leq C]=1$, for some $C < \infty$. Then
    \[\widehat{\tau}_{\delta} - \tau_{\delta} = \mathbb{P}_n[\dot{\tau}_{\delta}(O;P)] + O_P\left(\lVert \widehat{\pi} - \pi\rVert^2 + \lVert \widehat{\pi} - \pi\rVert\left\{\widehat{\mu}_0 - \mu_0\rVert + \lVert \widehat{\mu}_1 - \mu_1\rVert\right\}\right)+ o_P(n^{-1/2}),\]
    \[\widehat{\chi}_{\delta} - \chi_{\delta} = \mathbb{P}_n[\dot{\chi}_{\delta}(O;P)] + O_P\left(\lVert \widehat{\pi} - \pi\rVert^2\right)+ o_P(n^{-1/2}).\]
    In particular, if $\lVert \widehat{\pi} - \pi\rVert = o_P(n^{-1/4})$, then
    \[\sqrt{n}(\widehat{\chi}_{\delta} - \chi_{\delta}) \overset{d}{\to} \mathcal{N}(0, \mathrm{Var}\{\dot{\chi}_{\delta}(O;P)\}),\]
    and if additionally, $\lVert\widehat{\mu}_0 - \mu_0\rVert + \lVert \widehat{\mu}_1 - \mu_1\rVert = o_P(n^{-1/4})$, then
    \[\sqrt{n}(\widehat{\tau}_{\delta} - \tau_{\delta}) \overset{d}{\to} \mathcal{N}(0, \mathrm{Var}\{\dot{\tau}_{\delta}(O;P)\}),\]
    so that $\widehat{\tau}_{\delta}$ and $\widehat{\chi}_{\delta}$ are nonparametric efficient.
\end{theorem}

Putting our results together, the proposed estimated bounds in the outcome-based sensitivity models are
$\widehat{\tau}_{\delta} \pm \Gamma |e^{\delta} - 1|\widehat{\chi}_{\delta}$.
From Theorem~\ref{thm:est-bin-inc}, we can obtain asymptotically valid Wald-based confidence intervals for the lower and upper bounds. Concretely, a $(1 - \alpha)$-level confidence interval for the upper bound is given by
\[\{\widehat{\tau}_{\delta} + \Gamma|e^{\delta} - 1|\widehat{\chi}_{\delta}\} \pm z_{1 - \alpha/2} \sqrt{\frac{1}{n}\mathbb{P}_n\left[\left(\widehat{\psi}_{\tau_{\delta}} - \widehat{\tau}_{\delta} + \Gamma|e^{\delta} - 1|(\widehat{\psi}_{\chi_{\delta}} - \widehat{\chi}_{\delta})\right)^2\right]}.\]
The confidence interval for the lower bound is the same but with $\Gamma$ replaced by $- \Gamma$.

To end this section, we consider bound estimation for incremental effects in the marginal sensitivity model, i.e., Model~\ref{mod:sens-prop}. In particular, when $\delta > 0$, the bounds on $\mathbb{E}(Y(d_{q_{\delta}}^*))$ are \[\left[\tau_{\delta} - (e^{\delta} - 1)\zeta_{\delta,1}^-, \tau_{\delta} + (e^{\delta} - 1)\zeta_{\delta,1}^+\right],\]
and when $\delta < 0$, the bounds are
\[\left[\tau_{\delta} - (1 - e^{\delta})\zeta_{\delta,0}^-, \tau_{\delta} + (1 - e^{\delta})\zeta_{\delta,0}^+\right],\]
where
\[\zeta_{\delta,a}^- \coloneqq \mathbb{E}\left(\mathbb{E}(Y\{1 - g_a^-(X,Y)\} \mid X, A = a) \frac{\pi(1 - \pi)}{e^{\delta}\pi + (1 - \pi)}\right),\]
\[\zeta_{\delta,a}^+ \coloneqq \mathbb{E}\left(\mathbb{E}(Y\{g_a^+(X,Y) - 1\} \mid X, A = a) \frac{\pi(1 - \pi)}{e^{\delta}\pi + (1 - \pi)}\right),\]
for $a \in \{0,1\}$. Recalling that $g_a^- \equiv \Gamma^{\mathrm{sign}(\gamma_a^- - Y)}$ and $g_a^+ \equiv \Gamma^{\mathrm{sign}(Y - \gamma_a^+)}$ involve indicators for $Y$ exceeding or falling below its conditional quantiles, the functionals $\zeta_{\delta,a}^-$ and $\zeta_{\delta,a}^+$ are not pathwise differentiable. Nonetheless, we can still construct flexible, debiased estimators under relatively weak assumptions, just as in~\citet{dorn2024} who study very similar estimands. We give results for the upper bound when $\delta > 0$ here (i.e., an estimator of $\zeta_{\delta,1}^+$), and provide corresponding results for the lower bound and $\delta < 0$ in Appendix~\ref{app:estimation}. Whereas~\citet{dorn2024} require that $Y$ has a continuous distribution, given $X$ and $A$, with bounded density, we work with the following slightly weaker assumption.

\begin{assumption}[Margin condition 2]\label{ass:margin-2}
For $a \in \{0,1\}$, the CDF $y \mapsto P[Y \leq y \mid X, A =a]$ is strictly increasing at $\gamma_a^-(X), \gamma_a^+(X)$, and for some $b > 0, C > 0$,
    \[P[|Y - \gamma_a^-(X)| \leq t \mid X, A = a] \leq C \cdot t^{b}, \, P[|Y - \gamma_a^+(X)| \leq t \mid X, A = a] \leq C \cdot t^{b}, \text{ for all } t \geq 0,\]
with probability 1 (here, the constants $b$ and $C$ are independent of $X$ and $a$).
\end{assumption}

Letting $\kappa_1^+(X) = \mathbb{E}((Y - \gamma_1^+(X))_+ \mid X, A = 1)$, where $z_+ = \max{\{z,0\}}$ for $z \in \mathbb{R}$, we can rewrite the functional $\zeta_{\delta,1}^+$ as
\[\zeta_{\delta,1}^+ = \mathbb{E}\left(\left\{\left(1 - \frac{1}{\Gamma}\right)(\gamma_1^+ - \mu_1) + \left(\Gamma - \frac{1}{\Gamma}\right)\kappa_{1}^+\right\}\frac{\pi (1 - \pi)}{e^{\delta}\pi + (1 - \pi)}\right).\] On the basis of this parametrization, we propose the debiased estimator $\widehat{\zeta}_{\delta,1}^+ = \mathbb{P}_n[\widehat{\varphi}_{\delta,1}^+]$, where $\widehat{\varphi}_{\delta,1}^+$ is obtained by plugging in $(\widehat{\pi}, \widehat{\mu}_1, \widehat{\gamma}_1^+, \widehat{\kappa}_1^+)$ into the expression
\begin{align*}
    \varphi_{\delta,1}^+(O) &= \frac{1 - \pi}{e^{\delta}\pi + (1 - \pi)}\bigg\{ -\left(1 - \frac{1}{\Gamma}\right)\left\{\pi \mu_1 + A (Y - \mu_1)\right\} \\
    & \quad \quad \quad \quad \quad + \left(\Gamma - \frac{1}{\Gamma}\right)\left\{\pi\left(\frac{1}{1 + \Gamma} \gamma_1^+ + \kappa_1^+\right) + A\left[(Y - \gamma_1^+)_+ - \kappa_1^+\right]\right\}\bigg\} \\
    & \quad \quad + \frac{(1 - \pi)^2 - e^{\delta}\pi^2}{\{e^{\delta}\pi + (1 - \pi)\}^2}(A - \pi)\left\{\left(1 - \frac{1}{\Gamma}\right)(\gamma_1^+ - \mu_1) + \left(\Gamma - \frac{1}{\Gamma}\right)\kappa_1^+\right\}.
\end{align*}

\begin{theorem}\label{thm:est-bin-inc-prop}
    Assume that $\lVert \widehat{\pi} - \pi\rVert  + \lVert \widehat{\mu}_1 - \mu_1\rVert + \lVert \widehat{\gamma}_1^+ - \gamma_1^+\rVert + \lVert \widehat{\kappa}_1^+ - \kappa_1^+\rVert= o_P(1)$, and that $P[|Y| \leq C] = P[|\widehat{\mu}_1| \leq C]= P[|\widehat{\gamma}_1^+| \leq C] = P[|\widehat{\kappa}_1^+| \leq C] = 1$, for some $C < \infty$. Then under Assumption~\ref{ass:margin-2},
    \[\widehat{\zeta}_{\delta,1}^+ - \zeta_{\delta,1}^+ = (\mathbb{P}_n - \mathbb{E}_P) [\varphi_{\delta,1}^+(O)] + O_P(R_{\delta,1}^+)+ o_P(n^{-1/2}),\]
    where
    \[R_{\delta,1}^+ = \lVert \widehat{\pi} - \pi\rVert^2 + \lVert \widehat{\pi} - \pi\rVert\left\{\lVert \widehat{\mu}_1 - \mu_1\rVert + \lVert\widehat{\gamma}_1^+ - \gamma_1^+\rVert +\lVert\widehat{\kappa}_1^+ - \kappa_1^+\rVert\right\} + \lVert \widehat{\gamma}_{1}^+ - \gamma_1^+\rVert_{1 + b}^{1+b}.\]
    In particular, if $R_{\delta,1}^+ = o_P(n^{-1/2})$, then $\sqrt{n}(\widehat{\zeta}_{\delta,1}^+ - \zeta_{\delta,1}^+) \overset{d}{\to} \mathcal{N}(0, \mathrm{Var}\{\varphi_{\delta,1}^+(O)\})$.
\end{theorem}

Thus, when $\delta > 0$, the estimated sharp upper bound for $\mathbb{E}(Y(d_q^*))$ in Model~\ref{mod:sens-prop} is given by $\widehat{\tau}_{\delta} + (e^{\delta} - 1)\widehat{\zeta}_{\delta,1}^+$, and a Wald-based confidence interval (based on Theorems~\ref{thm:est-bin-inc} and \ref{thm:est-bin-inc-prop}) can be obtained as for previous estimators based on the empirical variance of $\widehat{\psi}_{\tau_{\delta}} + (e^{\delta} - 1)\widehat{\varphi}_{\delta,1}^+$.


\section{Discussion}\label{sec:discussion}
Here we showed that under violations of unconfoundedness, the causal effects of a policy---when allowed to depend on the natural value of treatment---are not characterized by the policy's induced treatment distribution. Importantly, the usual notion of a stochastic intervention (where treatment is assigned randomly on the basis only of covariate values) has the somewhat surprising property that bounds on its effects do not collapse to a point when the induced treatment distribution converges to that of the observational regime. In view of this fact, for the purpose of sensitivity analysis for unmeasured confounding, we propose instead to study what we term generalized policies, effectively unifying stochastic treatment rules and modified treatment policies. Although some similar policies have been conceptualized before (e.g., in~\citet{young2014} and~\citet{mauro2020}), they have not previously been studied in the sensitivity analysis context. Further, we develop a notion of optimality of a generalized policy among the class inducing the same treatment distribution, whereby the width of bounds on causal effects is minimized. For various sensitivity models, we characterize such optimal policies as well as sharp bounds on their causal effects. Finally, we develop flexible, efficient estimators of these sharp bounds for exponentially tilted target distributions. These statistical targets are non-standard and present new challenges, and the statistical methods developed in this work may be of independent interest beyond their use in causal sensitivity analysis. For instance, the bound functionals involve expected distances (e.g., total variation, Wasserstein) between conditional distributions, which are of interest more generally.

There are several avenues to pursue in future research, to build off this work. First, one could consider other sensitivity models beyond Models~\ref{mod:sens-bounded}--\ref{mod:sens-prop}, for instance the contamination model of~\citet{bonvini2022b} or the $L_2$-based propensity score model of~\citet{zhang2022}. Second, one could construct estimators of sharp bounds under treatment distributions beyond exponentially tilted targets. Third, it will be of interest to extend these methods to handle longitudinal settings with time-varying treatments. It is known, for instance, that point identification for longitudinal generalized policies requires stronger assumptions~\citep{young2014, diaz2023}. However, the extent to which our partial identification results  might extend to multiple timepoints remains an open question. Fourth, one possible trade-off we have not considered is that between bound length and efficiency in estimating the bounds themselves. That is, the ultimate causal effect uncertainty interval will depend on both the theoretical bound quantities and their confidence interval widths, and this could play an important role in finite samples. Lastly, it would be worth considering our ideas in the context of mediation analysis. Pure stochastic interventions have been proposed in the definition of new types of natural indirect and direct effects~\citep{vanderweele2014,diaz2024}, but analogues of generalized policies may yield interesting properties, especially for partial identification.

\section*{Acknowledgements}

EHK was supported by NSF CAREER Award 2047444.

\clearpage

\section*{References}
\vspace{-1cm}
\bibliographystyle{asa}
\bibliography{bibliography.bib}

\clearpage

\begin{appendices}

\section{Characterization of exponential tilts}\label{app:distance}

For the following results, we let $P$ be a probability measure on $(\mathbb{R}, \mathcal{B}(\mathbb{R}))$, and let $X \sim P$. Let $\Lambda_P \coloneqq \left\{\lambda \in \mathbb{R}: \mathbb{E}_P(e^{\lambda X}) < \infty\right\}$,  and define for any $\lambda \in \Lambda_P$ the tilted measure $Q_{\lambda}$ via
\[d Q_{\lambda}(x) = \frac{e^{\lambda x} \, dP(x)}{\mathbb{E}_P(e^{\lambda X})}.\]
Observe that $Q_{\lambda}$ and $P$ are mutually absolutely continuous, i.e., $Q_{\lambda} \ll P$ and $P \ll Q_{\lambda}$.

\begin{lemma}\label{lemma:KL-bound}
    Let $Q$ be any measure on $(\mathbb{R}, \mathcal{B}(\mathbb{R}))$ such that $\int x \, dQ(x) = m$. Then
    \[\mathrm{KL}(Q \mathrel{\Vert} P) \geq \sup_{\lambda \in \mathbb{R}} \left\{\lambda m - \log{\mathbb{E}_P(e^{\lambda X})}\right\}.\]
\end{lemma}

\begin{proof}
    Take $\lambda \in \mathbb{R}$, and we will show $\mathrm{KL}(Q \mathrel{\Vert} P) \geq \lambda m - \log{\mathbb{E}_P(e^{\lambda X})}$. It suffices to consider $Q \ll P$ and $\lambda \in \Lambda_P$ since otherwise the inequality holds trivially. Observe that
    \begin{align*}
        \mathrm{KL}(Q \mathrel{\Vert} P)
        &= \int \log{\left(\frac{dQ}{dP}\right)} \, dQ \\
        &= \int \log{\left(\frac{dQ}{dQ_{\lambda}}\frac{dQ_{\lambda}}{dP}\right)} \, dQ \\
        &= \mathrm{KL}(Q \mathrel{\Vert} Q_{\lambda}) + \int \log{\left(\frac{dQ_{\lambda}}{dP}\right)} \, dQ \\
        & \geq \int \log{\left(\frac{e^{\lambda x}}{\mathbb{E}_P(e^{\lambda X})}\right)} \, dQ(x) \\
        & = \lambda m - \log{\mathbb{E}_P(e^{\lambda X})},
    \end{align*}
    where the second line is justified since $Q \ll P \ll Q_{\lambda} \ll P$, and the fourth by the fact that KL-divergence is non-negative (which follows from Jensen's inequality)
\end{proof}

\begin{proposition}
    Consider the optimization problem over measures $Q$ given by
    \begin{align}
    \begin{split}\label{eq:KL-opt}
        \minimize_{Q} \ & \mathrm{KL}(Q \mathrel{\Vert} P) \\
        \mathrm{subject} \, \mathrm{to} & \int x \, dQ(x) = m.
    \end{split}
    \end{align}
    A solution is given by $Q_{\lambda_m}$, where $\lambda_m = \argmin\limits_{\lambda \in \mathbb{R}}{\{\mathbb{E}_P(e^{\lambda (X - m)})\}}$, 
    assuming $\{0, \lambda_m\} \subseteq \mathrm{int}{(\Lambda_P)}$.
\end{proposition}

\begin{proof}
    Observe that $\lambda m - \log{\mathbb{E}_P(e^{\lambda X})} \equiv -\log{\mathbb{E}_P(e^{\lambda(X - m)})}$, which is maximized at $\lambda_m$. Differentiating, we find that this occurs when $m = \frac{\mathbb{E}_P(Xe^{\lambda_m X})}{\mathbb{E}_P(e^{\lambda_m X})} = \mathbb{E}_{Q_{\lambda_m}}(X)$. Noting that $\mathrm{KL}(Q_{\lambda} \mathrel{\Vert} P) = \lambda\mathbb{E}_{Q_{\lambda}}(X) - \log{\mathbb{E}_P(e^{\lambda X})}$, it follows by Lemma~\ref{lemma:KL-bound} that $Q_{\lambda_m}$ is a solution to~\eqref{eq:KL-opt}. More formally, define $g(\lambda) = \mathbb{E}_P(e^{\lambda(X - m)})$, and observe that by the dominated convergence theorem and the assumption that $0 \in \mathrm{int}{(\Lambda_P)}$,
    \[\frac{d}{d\lambda}g(\lambda) = \mathbb{E}_P((X - m)e^{\lambda(X - m)}); \ \frac{d^2}{d\lambda^2}g(\lambda) = \mathbb{E}_P((X - m)^2 e^{\lambda(X - m)}),\]
    for all $\lambda \in \mathrm{int}{(\Lambda_P)}$.
    This shows that $g$ is strictly convex (unless $P$ is a point mass at $m$ so that the statement holds trivially), with (unique) minimizer $\lambda_m$ assuming $\lambda_m \in \mathrm{int}{(\Lambda_P)}$.
\end{proof}

\section{Maximal couplings}\label{app:maximal}

\begin{theorem}\label{thm:TV-max}
    Let $P$ and $Q$ be probability measures on $(\mathbb{R}, \mathcal{B}(\mathbb{R}))$, and let $\mathcal{J}(P, Q)$ be the set of joint distributions with marginals $P$ and $Q$. Then
    \[\mathbb{P}_{(X, Y) \sim J}[X \neq Y] \geq \mathrm{TV}(P, Q), \text{ for all } J \in \mathcal{J}(P, Q),\]
    where $\mathrm{TV}(P, Q) = \sup_{B \in \mathcal{B}(\mathbb{R})}|P(B) - Q(B)|$ is the total variation distance between $P$ and $Q$. Moreover, there exists $J^* \in \mathcal{J}(P, Q)$ such that $\mathbb{P}_{(X, Y) \sim J^*}[X \neq Y] = \mathrm{TV}(P, Q)$.
\end{theorem}

\begin{proof}
    Consider any $J \in \mathcal{J}(P, Q)$; we will write 
    \[J[X \in B_1, Y \in B_2] \coloneqq \mathbb{P}_{(X, Y) \sim J}[X \in B_1, Y \in B_2],\] for any $B_1, B_2 \in \mathcal{B}(\mathbb{R})$. Letting $B \in \mathcal{B}(\mathbb{R})$ be arbitrary, notice that
    \begin{align*}
        & P(B) - Q(B) = J[X \in B] - J[Y \in B] \\
        & = J[X \in B, X = Y] + J[X \in B, X \neq Y] - J[Y \in B, X = Y] - J[Y \in B, X \neq Y] \\
        &= J[X \in B, X \neq Y] - J[Y \in B, X \neq Y]  \leq J[X \neq Y],
    \end{align*}
    since $J[X \in B, X = Y] = J[X \in B, X = Y, Y \in B] = J[Y \in B, X = Y]$. By a symmetric argument, we also have $Q(B) - P(B) \leq J[X \neq Y]$ which implies $\mathrm{TV}(P, Q) \leq J[X \neq Y]$ after taking a supremum over $B \in \mathcal{B}(\mathbb{R})$. This proves the first statement of the theorem.
    
    Next, to construct the (maximal) coupling $J^*$ for which equality is attained, we will proceed similarly to Example~\ref{ex:maximal}. Let $\rho$ be a dominating measure for $P$ and $Q$ (e.g., $\rho = \frac{1}{2}(P + Q)$), and define $p = \frac{dP}{d\rho}$, $q = \frac{dQ}{d\rho}$. We now define two measures $J_1^*, J_2^*$ on $(\mathbb{R}^2, \mathcal{B}(\mathbb{R}^2))$ as follows. First, we define $J_1^* = \nu \circ \phi^{-1}$, where $\phi: \mathbb{R} \to \mathbb{R}^2$ is the map $x \mapsto (x,x)$, and
    \[\nu(B) = \int_B \min{\{p(x), q(x)\}} \, d\rho(x), \text{ for any } B \in \mathcal{B}(\mathbb{R}).\]
    Note that $\nu$ is a well-defined finite measure on $(\mathbb{R}, \mathcal{B}(\mathbb{R}))$, with $\nu(\mathbb{R}) = 1 - \mathrm{TV}(P, Q)$ by Lemma~\ref{lemma:TV-eq} below. Since $\phi$ is a measurable map from $(\mathbb{R}, \mathcal{B}(\mathbb{R}))$ to $(\mathbb{R}^2, \mathcal{B}(\mathbb{R}^2))$, it follows that $J_1^*$ is a finite measure on $(\mathbb{R}^2, \mathcal{B}(\mathbb{R}^2))$, supported on the diagonal $C_{\mathrm{diag}} = \{(x, x) : x \in \mathbb{R}\}$.
    Next, if $\mathrm{TV}(P, Q) = 0$ then set $J_2^* \equiv 0$, otherwise define the measure $J_2^*$ via
    \[J_2^*(C) = \frac{1}{\mathrm{TV}(P, Q)}\int_C (p(x) - q(x))_+ (q(y) - p(y))_+ \, d\rho^2(x \otimes y), \text{ for any } C \in \mathcal{B}(\mathbb{R}^2),\]
    where $(z)_+ = \max{\{z, 0\}}$ for $z \in \mathbb{R}$, which is a well-defined finite measure with $J_2^*(\mathbb{R}^2) = \mathrm{TV}(P, Q)$, by Lemma~\ref{lemma:TV-eq} and Fubini's theorem. Thus, $J^* \coloneqq J_1^* + J_2^*$ is a probability measure on $(\mathbb{R}^2, \mathcal{B}(\mathbb{R}^2))$, and has the desired marginal distributions: for any $B \in \mathcal{B}(\mathbb{R})$,
    \begin{align*}
        J^*(B \times \mathbb{R}) &= \nu \circ \phi^{-1}(B \times \mathbb{R}) + \frac{1}{\mathrm{TV}(P, Q)}\int_B (p(x) - q(x))_+ \, d\rho(x) \int_{\mathbb{R}} (q(y) - p(y))_+ \, d\rho(y) \\
        &= \int_B \min{\{p(x), q(x)\}}\, d\rho(x) + \int_B (p(x) - q(x))_+ \, d\rho(x) \\
        &= \int_B p(x) \, d\rho(x) = P(B),
    \end{align*}
    and similarly $J^*(\mathbb{R} \times B) = \int_B q\, d\rho = Q(B)$, 
    where we used Lemma~\ref{lemma:TV-eq} and the fact that $\phi^{-1}(B \times \mathbb{R}) = \phi^{-1}(\mathbb{R} \times B)= B$. Finally, $J^*[X = Y] \equiv J^*(C_{\mathrm{diag}}) \geq J_1^*(C_{\mathrm{diag}}) = 1 - \mathrm{TV}(P, Q)$, so that $J^*[X \neq Y] \leq \mathrm{TV}(P, Q)$, and equality holds by the first part of the theorem.
    
\end{proof}

\begin{lemma}\label{lemma:TV-eq}
    Given the notation in the proof of Theorem~\ref{thm:TV-max},
    \[\mathrm{TV}(P, Q) = \frac{1}{2}\int |p(x) - q(x)|\, d\rho(x)= 1 - \int \min{\{p(x), q(x)\}} \, d\rho(x),\]
    and consequently,
    $\mathrm{TV}(P, Q) = \int (p(x) - q(x))_+ \, d\rho(x) = \int (q(y) - p(y))_+ \, d\rho(y)$.
\end{lemma}

\begin{proof}
    Let $A \in \mathcal{B}(\mathbb{R})$ be defined by $A = \{x \in \mathbb{R}: p(x) > q(x)\}$. Since $\int p \, d\rho = \int q \, d \rho = 1$,
    \begin{align*}
        \int |p(x) - q(x)| \, d\rho(x)
        &= \int_A (p - q) \, d\rho + \int_{A^c} (q - p) \, d\rho \\
        &= \int_A p \, d\rho + \int_{A^c} q \, d\rho -\int \min{\{p, q\}} \, d\rho \\
        &= 2 - \int_{A^c} p \, d\rho - \int_{A} q \, d\rho - \int \min{\{p, q\}}\, d\rho \\
        &= 2\left(1 -\int \min{\{p, q\}}\, d\rho\right).
    \end{align*}
    Notice also that
    \begin{align*}
        P(A) - Q(A) &= \frac{1}{2}\{P(A) - Q(A) + Q(A^c) - P(A^c)\} \\
        &= \frac{1}{2}\left\{\int_A (p - q) \, d\rho + \int_{A^c} (q-p) \, d\rho\right\} \\
        &= \frac{1}{2}\int |p - q| \, d\rho.
    \end{align*}
    It remains to show $\mathrm{TV}(P, Q) = P(A) - Q(A)$: for any $B \in \mathcal{B}(\mathbb{R})$,
    \begin{align*}
        P(B) - Q(B) 
        &= \int_B p(x) \, d\rho(x) - \int_B q(x) \, d\rho(x) \\
        &= \int_{B\cap A} p(x) \, d\rho(x) + \int_{B\setminus A} p(x) \, d\rho(x)- \int_{B\cap A} q(x) \, d\rho(x) - \int_{B\setminus A} q(x) \, d\rho(x) \\
        & \leq 
        \int_{B\cap A} p(x) \, d\rho(x) - \int_{B\cap A} q(x) \, d\rho(x) \\
        &\leq \int_{B\cap A} p(x) \, d\rho(x) - \int_{B\cap A} q(x) \, d\rho(x) + \int_{A \setminus B} p(x) \, d\rho(x) - \int_{A \setminus B} q(x) \, d\rho(x) \\
        &= P(A) - Q(A)
    \end{align*}
    and, by symmetry, $Q(B) - P(B) \leq Q(A^c)- P(A^c) = P(A) - Q(A)$. Taking a supremum over $B \in \mathcal{B}(\mathbb{R})$, this shows that $\mathrm{TV}(P, Q) \leq P(A) - Q(A)$, but we must have equality since clearly $P(A) -Q(A) \leq \sup_{B \in \mathcal{B}(\mathbb{R})}|P(B) - Q(B)|$.
\end{proof}

Note that, by our results,
\[\mathrm{TV}(P, Q) = \inf{\left\{J[X \neq Y] : J \in \mathcal{J}(P, Q)\right\}}.\] In the proof of Theorem~\ref{thm:TV-max}, we constructed a $J^*$ that attained this infimum. However, this construction is not unique. In the remainder of this appendix, we show that the construction in Example~\ref{ex:maximal}, and that for binary treatments shown in Section~\ref{sec:binary}, also imply valid maximal couplings (i.e., attain the infimum above).

\begin{proof}[Verification of construction in Example~\ref{ex:maximal}]
    Note that \[P[d_Q^* = A \mid X] \geq P[V_1 \leq \widetilde{p}(A \mid X) \mid X] = \int \min{\{\pi(a \mid X), q(a \mid X)\}} \, d\rho(a) = 1 - \mathrm{TV}_{\Pi, Q}(X),\]
    by construction and Lemma~\ref{lemma:TV-eq}. Thus, $P[d_Q^* \neq A \mid X] \leq \mathrm{TV}_{\Pi, Q}(X)$, and equality must hold by Theorem~\ref{thm:TV-max}. It remains to verify that the distribution of $d_Q^*$ given $X$ is indeed $Q$. To see this, observe first that $\widetilde{q}(\, \cdot \mid X)$ is a valid density, again by Lemma~\ref{lemma:TV-eq}, and moreover that $P[d_Q^* = A, V_1 > \widetilde{p}(A \mid X)\mid X] = 0$, since otherwise we would contradict Theorem~\ref{thm:TV-max}. Thus, for any $B \in \mathcal{B}(\mathbb{R})$,
    \begin{align*}
    P[d_Q^* \in B, d_Q^* = A \mid X] 
    &= P[A \in B, V_1 \leq \widetilde{p}(A \mid X) \mid X]\\
    &= \int_B \widetilde{p}(A \mid X) \pi(a \mid X)\, d\rho(a) \\
    &= \int_B \min{\{\pi(a \mid X), q(a \mid X)\}} \, d\rho(a),
    \end{align*} 
    where we used Assumption~\ref{ass:aux}. Similarly,
    \begin{align*}
    P[d_Q^* \in B, d_Q^* \neq A \mid X] 
    &= P[d_Q^* \in B, V_1 > \widetilde{p}(A \mid X) \mid X]\\
    &= P[Q^{-1}(V_2 \mid X) \in B, V_1 > \widetilde{p}(A \mid X) \mid X] \\
    &= \mathrm{TV}_{\Pi, Q}(X) \int_B \widetilde{q}(a \mid X) \, d\rho(a) \\
    &= \int_B \big[q(a \mid X) - \min{\{\pi(a \mid X), q(a \mid X)\}}\big] \, d\rho(a),
    \end{align*} 
    where we used Assumption~\ref{ass:aux}, the independence of $V_1$ and $V_2$, and Lemma~\ref{lemma:TV-eq}. Combining these results, we conclude that $P[d_Q^* \in B \mid X] = \int_B q(a \mid X) \, d\rho(a)$, so that 
    $d_Q^* \mid X \sim Q(\, \cdot \mid X)$.
\end{proof}

\begin{proof}[Verification of construction in Section~\ref{sec:binary}]
    Lemma~\ref{lemma:TV-eq} implies that $\mathrm{TV}_{\Pi,Q}(X) \equiv |\pi(X) - q(X)|$ for binary treatments. Now, we can directly verify that $P[A \neq d_q^*\mid X] = |\pi(X) - q(X)|$: (i) if $\pi(X) = q(X)$, then $P[A \neq d_q^* \mid X] = 0 = \mathrm{TV}_{\Pi, Q}(X)$, (ii) if $\pi(X) > q(X)$, then
    \[P[A \neq d_q^* \mid X] = P[A = 1 \mid X] \cdot P\left[V > \frac{q(X)}{\pi(X)} \, \bigg| \, X \right] = \pi(X) \cdot \left(1 - \frac{q(X)}{\pi(X)}\right) = |\pi(X) - q(X)|,\]
    and (iii) if $\pi(X) < q(X)$, then
    {\small
    \[P[A \neq d_q^* \mid X] = P[A = 0 \mid X] \cdot P\left[V > \frac{1 - q(X)}{1 - \pi(X)} \, \bigg| \, X \right] = (1 - \pi(X))\cdot \left(1 - \frac{1 - q(X)}{1 - \pi(X)}\right) = |\pi(X) - q(X)|.\]
    }%
    It is similarly straightforward to verify that $P[d_q^* = 1 \mid X] = q(X)$, so we are done.
\end{proof}

\section{Sensitivity models}\label{app:sens}

\begin{proof}[Proof of properties in Model~\ref{mod:sens-prop}]
Observe that, applying Bayes' rule twice,
    \begin{align*}
        \mathbb{E}(Y(a) \mid X, A = a')
        &= \int y \, dP_{Y(a) \mid X, A = a'}^*(y \mid X, a') \\
        &= \int y \frac{\eta_a(a' \mid X, y)}{\pi(a' \mid X)} \, dP_{Y(a) \mid X}^*(y \mid X) \\
        &= \int y \frac{\eta_a(a' \mid X, y)}{\pi(a' \mid X)} \frac{\pi(a \mid X)}{\eta_a(a\mid X, y)} \, dP_{Y \mid X, A = a}(y \mid X, a)\\
        &= \mathbb{E}\left(Y \frac{\eta_a(a' \mid X, Y)}{\pi(a' \mid X)} \frac{\pi(a \mid X)}{\eta_a(a\mid X, Y)}\,\, \middle| \,\, X, A = a\right),
    \end{align*}
    where we also invoked Assumption~\ref{ass:consistency}. Similarly, reversing these steps,
    \begin{align*}
        \mathbb{E}\left(\frac{\eta_a(a' \mid X, Y)}{\pi(a' \mid X)} \frac{\pi(a \mid X)}{\eta_a(a\mid X, Y)}\,\, \middle| \,\, X, A = a\right)
        &= \int \, dP_{Y(a) \mid X, A = a'}^*(y \mid X, a') = 1.
    \end{align*}

    Now, consider sharp bounds on $\nu_a(X,a')$ in this model. For the lower bound, we aim to solve the following minimization problem: for $a, a' \in \mathcal{A}$ fixed,
    \begin{align}
    \begin{split}\label{eq:min-prop}
        \minimize_{h(a' \mid X, a, \cdot)} \; & \mathbb{E}(Y h(a' \mid X, a, Y) \mid X, A = a) \\
        \text{subject to } \; & \mathbb{E}(h(a' \mid X, a, Y) \mid X, A = a)) = 1 \\
        & \Gamma^{-1} \leq h(a' \mid X, a, Y) \leq \Gamma, \text{ for all } a' \in \mathcal{A}
    \end{split}
    \end{align}
    This is a linear program, and we can characterize its solution using the Karush–Kuhn–Tucker (KKT) conditions. Since the problem is trivial at $\Gamma = 1$, we will assume that $\Gamma > 1$. Setting the (functional) derivative of the Lagrangian to zero results in the stationarity condition
    \[0 = Y - \lambda_1^{X, Y} + \lambda_2^{X, Y} - \lambda_3^X,\]
    for Lagrange multipliers $\lambda_1^{X, Y},\lambda_2^{X, Y},\lambda_3^{X}$. The other conditions are (i) complementary slackness, i.e., that $\lambda_1^{X, Y}(\Gamma^{-1} - h(a' \mid X, a, Y)) = 0$ and $\lambda_2^{X, Y}(h(a' \mid X, a, Y) - \Gamma) = 0$; (ii) primal feasibility, i.e., that the constraints in~\eqref{eq:min-prop} are satisfied; and (iii) dual feasibility, i.e., that $\lambda_1^{X, Y} \geq 0$ and $\lambda_2^{X, Y} \geq 0$. If $\lambda_1^{X, Y} \neq 0$, then $h(a' \mid X, a, Y) = \Gamma^{-1}$, implying $\lambda_2^{X, Y} = 0$ by complementary slackness, and $\lambda_1^{X, Y} = Y - \lambda_3^X \implies Y \geq \lambda_3^X$ by dual feasibility. Similarly, If $\lambda_2^{X, Y} \neq 0$, then $h(a' \mid X, a, Y) = \Gamma$, implying $\lambda_1^{X, Y} = 0$ by complementary slackness, and $\lambda_2^{X, Y} = - Y + \lambda_3^X \implies Y \leq \lambda_3^X$ by dual feasibility. Combining these cases with the primal feasibility condition, we must have
    \[1 = \Gamma^{-1} P[Y > \lambda_3^X \mid X, A = a] + \Gamma P[Y < \lambda_3^X \mid X, A = a] + \lambda_3^X P[Y = \lambda_3^X \mid X, A = a].\]
    In particular, if the distribution function $y \mapsto P[Y \leq y \mid X, A = a]$ is strictly increasing at $\gamma_a^-(X)$ and $P[Y < \gamma_a^-(X) \mid X, A = a] = P[Y \leq \gamma_a^-(X) \mid X, A = a] = \frac{1}{1 + \Gamma}$, then the above equation is satisfied by taking $\lambda_3^X = \gamma_a^-(X)$. This yields the optimum $h^*(a' \mid X, a, Y) \equiv g_a^-(X,Y) = \Gamma^{\mathrm{sign}(\gamma_a^-(X) - Y)}$. The upper bound is argued in the same way, replacing $Y$ with $-Y$, resulting in the optimum $g_a^+(X,Y) = \Gamma^{\mathrm{sign}(Y - \gamma_a^+(X))}$ where $\gamma_a^+(X)$ is the $\frac{\Gamma}{1 + \Gamma}$ quantile of $Y$ given $X$ and $A = a$ (assuming the distribution function is strictly increasing at $\gamma_a^+(X,Y)$).
\end{proof}

\begin{proof}[Proof of Proposition~\ref{prop:ident}]
See that
\begin{align*}
    \mathbb{E}(Y(d) \mid X)
    &= \mathbb{E}(\mathbb{E}(Y(d) \mid X, A, d) \mid X) \\
    &= \int \mathbb{E}(Y(a) \mid X, A = a', d = a) \, dP_{A, d \mid X}(a', a \mid X) \\
    &= \int \nu_a(X, a') \, dP_{A, d \mid X}(a', a \mid X) \\
    &= \tau_Q(X) + \int \{\nu_a(X, a') -\mu_a(X)\}\, dP_{A, d \mid X}(a', a \mid X) \\
    &= \tau_Q(X) + \int \mathds{1}(a' \neq a)\{\nu_a(X, a') -\mu_a(X)\}\, dP_{A, d \mid X}(a', a \mid X).
\end{align*}
The first equality results from the tower law, the third by Assumption~\ref{ass:aux} (i.e., $d(X, A, V) \ind Y(a) \mid X, A$), the fourth by rearranging, and the fifth by Assumption~\ref{ass:consistency}.
\end{proof}

\begin{proof}[Proof of Theorem~\ref{thm:maximal}]
    The first part of the statement (i.e., properties of the maximal coupling) is proven in Appendix~\ref{app:maximal}. For the second, observe that by Proposition~\ref{prop:ident},
    \begin{align*}
        -\Gamma^-(X) P[A \neq d_Q^* \mid X]\leq \mathbb{E}(Y(d_Q^*)\mid X) - t_Q(X) \leq \Gamma^+(X) P[A \neq d_Q^* \mid X],
    \end{align*}
    which implies the stated bounds as $P[A \neq d_Q^* \mid X] = \mathrm{TV}_{\Pi, Q}(X)$. These are the narrowest bounds of this kind due to the maximality property of $d_Q^*$: $P[A \neq d_Q^* \mid X] \leq P[A \neq d \mid X]$ for all $d \in \mathcal{D}_Q$. Finally, if $\Gamma^-, \Gamma^+$ are tight/sharp, then so are the inequalities above (and hence the bounds for $\mathbb{E}(Y(d_Q^*)\mid X)$), as we have just used the equality proven in Proposition~\ref{prop:ident}.
\end{proof}

\begin{proof}[Proof of Theorem~\ref{thm:monotonic}]
    This is an immediate corollary of the classical characterization of the one-dimensional solution to the Monge-Kantorovich problem for convex cost. See, for example, Theorem 2.9 in \citet{santambrogio2015}. Combining this with the same argument as the proof of Theorem~\ref{thm:maximal} yields the result.
\end{proof}

\begin{proof}[Bounds for maximal $Q$-policy in Model~\ref{mod:sens-prop}]
    We have shown for Model~\ref{mod:sens-prop} that
    \[-\Gamma_a^-(X) \leq \nu_a(X, a') - \mu_a(X) \leq \Gamma_a^+(X),\]
    where
    \[\Gamma_a^-(X) = \mathbb{E}(Y \{1 - g_a^-(X, Y)\} \mid X, A = a), \, \Gamma_a^+(X) = \mathbb{E}(Y \{g_a^+(X, Y) - 1\}\mid X, A = a).\] By Proposition~\ref{prop:ident}, we have
    \[-\mathbb{E}(\mathds{1}(A \neq d_Q^*) \Gamma_{d_Q^*}^-(X) \mid X)\leq \mathbb{E}(Y(d_Q^*) \mid X) - t_q(X) \leq \mathbb{E}(\mathds{1}(A \neq d_Q^*) \Gamma_{d_Q^*}^+(X) \mid X).\]
    By our results on maximal policies, we know that
    \begin{align*}
        \mathbb{E}(\mathds{1}(A \neq d_Q^*) \Gamma_{d_Q^*}^+(X) \mid X) 
        &= \mathbb{E}(\mathds{1}(V_1 > \widetilde{p}(A\mid X)) \Gamma_{\widetilde{Q}^{-1}(V_2 \mid X)}^+(X) \mid X) \\
        &= \mathrm{TV}_{\Pi, Q}(X) \int \Gamma_a^+(X) \widetilde{q}(a \mid X) \, d\rho(a) \\
        &= \int \Gamma_a^+(X) (q(a \mid X) - \pi(a \mid X))_+ \, d\rho(a).
    \end{align*}
    Arguing similarly for the lower bound, we conclude that
    \begin{align*}
        \mathbb{E}(Y(d_Q^*) \mid X) \in &\bigg[t_Q(X) - \int \Gamma_a^-(X) (q(a \mid X) - \pi(a \mid X))_+ \, d\rho(a), \\
        & \quad \quad t_Q(X) + \int \Gamma_a^+(X) (q(a \mid X) - \pi(a \mid X))_+ \, d\rho(a)\bigg].
    \end{align*}
\end{proof}

\begin{proof}[Proof of Theorem~\ref{thm:maximal-bin}]
    By Corollary~\ref{cor:ident-bin} and inequalities~\eqref{eq:ineq-maximal-bin}, for any $d \in \mathcal{D}_q$,
    \[-\Gamma_0^-\pi(1 - \widetilde{q}_d) - \Gamma_1^- (q - \pi\widetilde{q}_d)\leq \mathbb{E}(Y(d)\mid X) - t_q \leq \Gamma_0^+\pi(1 - \widetilde{q}_d) + \Gamma_1^+ (q - \pi\widetilde{q}_d),\]
    and the width of these bounds is
    \[(\Gamma_0^- + \Gamma_0^+)\pi(1 - \widetilde{q}_d) + (\Gamma_1^- + \Gamma_1^+)(q - \pi \widetilde{q}_d).\]
    The width is decreasing in $\widetilde{q}_d$, so minimizing this objective subject to the constraints $\widetilde{q}_d \in [0,1]$ and $\pi \widetilde{q}_d \in [0,q]$ (which is required so that the implied $P[d = 1, A = 0 \mid X] = q - \pi \widetilde{q}_d \in [0,q]$), we find that the minimum is achieved when $\widetilde{q}_d = \min{\{1, q/\pi\}}$. But this corresponds to the maximal coupling: one can verify by comparing the resulting distribution with that of $(A, d_q^*)$, where $d_q^*$ is as defined in Section~\ref{sec:binary}. Further, plugging in $\widetilde{q}_d = \min{\{1, q/\pi\}}$ into the bounds above, we find
    \[-(1 - s)\Gamma_0^-(\pi - q) - s\Gamma_1^- (q - \pi)\leq \mathbb{E}(Y(d_q^*)\mid X) - t_q \leq (1 - s)\Gamma_0^+(\pi - q) + s\Gamma_1^+ (q - \pi),\]
    where $s = \mathds{1}(q > \pi)$. Equivalently, the lower and upper bounds can be written $-\Gamma_s^- |q - \pi|$ and $\Gamma_s^+ |q-\pi|$, respectively, which yields the result. The claim about tightness/sharpness is immediate, as if the inequalities in~\eqref{eq:ineq-maximal-bin} are best possible (or attainable), then so are the bounds on $\mathbb{E}(Y(d_q^*) \mid X)$.
\end{proof}

\section{Estimation results}\label{app:estimation}

\begin{proof}[Proof of Proposition~\ref{prop:von-Mises}]
    Observe that, by iterated expectations,
    \begin{align*}
        \tau_{\delta}(\overline{P}) - \tau_{\delta}(P) + \mathbb{E}_P(\dot{\tau}_{\delta}(O;P))
        &= \mathbb{E}_P\left(\frac{e^{\delta A}}{\overline{\alpha}_{\delta}}(Y - \overline{\gamma}_{\delta}) + \overline{\gamma}_{\delta} - \gamma_{\delta}\right) \\
        &= \mathbb{E}_P\left(\frac{1}{\overline{\alpha}_{\delta}}(\gamma_{\delta}\alpha_{\delta} - \overline{\gamma}_{\delta}\alpha_{\delta}) + \overline{\gamma}_{\delta} - \gamma_{\delta}\right) \\
        &= \mathbb{E}_P\left(\{\overline{\gamma}_{\delta} - \gamma_{\delta}\}\left\{1 - \frac{\alpha_{\delta}}{\overline{\alpha}_{\delta}}\right\}\right),
    \end{align*}
    as claimed.
\end{proof}

\begin{proof}[Proof of Theorem~\ref{thm:est-cont-inc}]
    Throughout, we write $P(f) = \int f(o) \, dP(o)$, where $f$ can depend on training data used for nuisance estimation. We begin the error expansion of $\widehat{\tau}_{\delta}$:
    \begin{align*}
        \widehat{\tau}_{\delta} - \tau_{\delta} &= \mathbb{P}_n[\widehat{\Psi}_{\tau_{\delta}}] - P(\Psi_{\tau_{\delta}}) \\
        &= (\mathbb{P}_n - P)\Psi_{\tau_{\delta}} + P(\widehat{\Psi}_{\tau_{\delta}} - \Psi_{\tau_{\delta}}) + (\mathbb{P}_n - P)(\widehat{\Psi}_{\tau_{\delta}} - \Psi_{\tau_{\delta}}),
    \end{align*}
    where $\Psi_{\tau_{\delta}} = \frac{e^{\delta A}}{\alpha_{\delta}(X)}(Y - \gamma_{\delta}(X)) + \gamma_{\delta}(X)$.
    The first summand is the same as $\mathbb{P}_n[\dot{\tau}_{\delta}(O;P)]$, and the second summand is $O_P(R_{1,n})$ by Proposition~\ref{prop:von-Mises} and the assumption that $\widehat{\alpha}_{\delta}$ is bounded below. It remains to control the third (empirical process) term. 
    
    By Lemma 2 in~\citet{kennedy2020b}, to show that this term is $o_P(n^{-1/2})$ it suffices to show that $\lVert \widehat{\Psi}_{\tau_{\delta}} - \Psi_{\tau_{\delta}}\rVert = o_P(1)$; this is due to the fact that $\widehat{\Psi}_{\tau_{\delta}}$ is fit on separate independent data and an application of Markov's inequality. To this end, observe that
    \begin{align*}
        \lVert \widehat{\Psi}_{\tau_{\delta}} - \Psi_{\tau_{\delta}}\rVert
        &= \left\lVert\frac{e^{\delta A}}{\widehat{\alpha}_{\delta}}(Y - \widehat{\gamma}_{\delta}) + \widehat{\gamma}_{\delta} - \frac{e^{\delta A}}{\alpha_{\delta}}(Y - \gamma_{\delta}) + \gamma_{\delta}\right\rVert \\
        & \leq \left\lVert \frac{e^{\delta} A}{\widehat{\alpha}_{\delta}}(\gamma_{\delta} - \widehat{\gamma}_{\delta})\right\rVert + \left\lVert \frac{e^{\delta A}}{\widehat{\alpha}_{\delta}\alpha_{\delta}}(\alpha_{\delta} - \widehat{\alpha}_{\delta})(Y - \gamma_{\delta})\right\rVert + \lVert \widehat{\gamma}_{\delta} - \gamma_{\delta}\rVert \\
        & \lesssim \lVert \widehat{\alpha}_{\delta} - \alpha_{\delta}\rVert + \lVert \widehat{\gamma}_{\delta} - \gamma_{\delta}\rVert \\
        &= o_P(1),
    \end{align*}
    by our consistency and boundedness assumptions.

    For $\widehat{\xi}_{\delta}$, we have a similar error decomposition:
    \begin{align*}
        \widehat{\xi}_{\delta} - \xi_{\delta} &= \mathbb{P}_n[\widehat{\Psi}_{\xi_{\delta}}] - P(\Psi_{\xi_{\delta}}) \\
        &= (\mathbb{P}_n - P)\Psi_{\xi_{\delta}} + P(\widehat{\Psi}_{\xi_{\delta}} - \Psi_{\xi_{\delta}}) + (\mathbb{P}_n - P)(\widehat{\Psi}_{\xi_{\delta}} - \Psi_{\xi_{\delta}}),
    \end{align*}
    For the second (asymptotic bias) term,
    \begin{align*}
        & P(\widehat{\Psi}_{\xi_{\delta}} - \Psi_{\xi_{\delta}}) \\
        &= P\left(\left\{\mathds{1}(e^{\delta A} < \widehat{\alpha}_{\delta}) - \widehat{\kappa}_{\delta}\right\}\left\{1 - \frac{e^{\delta A}}{\widehat{\alpha}_{\delta}}\right\} - \mathds{1}(e^{\delta A} < \alpha_{\delta}) + \kappa_{\delta}\right) \\
        &= P\left(\left\{\mathds{1}(e^{\delta A} < \widehat{\alpha}_{\delta}) - \mathds{1}(e^{\delta A} < \alpha_{\delta})\right\}\left\{1 - \frac{e^{\delta A}}{\widehat{\alpha}_{\delta}}\right\} -  \widehat{\kappa}_{\delta}\left\{1 - \frac{e^{\delta A}}{\widehat{\alpha}_{\delta}}\right\} - \mathds{1}(e^{\delta A} < \alpha_{\delta})\frac{e^{\delta} A}{\widehat{\alpha}_{\delta}} + \kappa_{\delta}\right) \\
        &= P\left(\frac{1}{\widehat{\alpha}_{\delta}}\left\{\mathds{1}(e^{\delta A} < \widehat{\alpha}_{\delta}) - \mathds{1}(e^{\delta A} < \alpha_{\delta})\right\}\left\{\widehat{\alpha}_{\delta} - e^{\delta A} \right\}\right) + P\left(\{\kappa_{\delta} - \widehat{\kappa}_{\delta}\}\left\{1 - \frac{\alpha_{\delta}}{\widehat{\alpha}_{\delta}}\right\}\right).
    \end{align*}
    The second summand is $O_P(R_{2,n})$, while the first is controlled under Assumption~\ref{ass:margin}:
    \begin{align*}
        & \left|P\left(\frac{1}{\widehat{\alpha}_{\delta}}\left\{\mathds{1}(e^{\delta A} < \widehat{\alpha}_{\delta}) - \mathds{1}(e^{\delta A} < \alpha_{\delta})\right\}\left\{\widehat{\alpha}_{\delta} - e^{\delta A} \right\}\right)\right| \\
        & \lesssim P\left(\mathds{1}(\max{\{|\alpha_{\delta} - e^{\delta A}|, |\widehat{\alpha}_{\delta} - e^{\delta A}|\}} \leq |\widehat{\alpha}_{\delta} - \alpha_{\delta}|) \cdot |\widehat{\alpha}_{\delta} - e^{\delta A}|\right) \\
        & \lesssim \lVert\widehat{\alpha}_{\delta} - \alpha_{\delta}\rVert_{\infty}^{1 + b} = R_{3,n},
    \end{align*}
    where the second line follows from Lemma 4 in~\citet{levis2024} or Lemma 1 in~\citet{kennedy2020b}, and the third line follows by the margin condition. It remains to show that the empirical process term is $o_P(n^{-1/2})$, which again will follow from showing that $\lVert \widehat{\Psi}_{\xi_{\delta}} - \Psi_{\xi_{\delta}}\rVert = o_P(1)$, by Lemma 2 of~\citet{kennedy2020b}. See that under our boundedness assumptions,
    \begin{align*}
        \lVert \widehat{\Psi}_{\xi_{\delta}} - \Psi_{\xi_{\delta}}\rVert
        &= \left\lVert \left\{\mathds{1}(e^{\delta A} < \widehat{\alpha}_{\delta}) - \widehat{\kappa}_{\delta}\right\}\left\{1 - \frac{e^{\delta A}}{\widehat{\alpha}_{\delta}}\right\} - \left\{\mathds{1}(e^{\delta A} < \alpha_{\delta}) - \kappa_{\delta}\right\}\left\{1 - \frac{e^{\delta A}}{\alpha_{\delta}}\right\}\right\rVert \\
        & \lesssim \left\lVert \left\{\mathds{1}(e^{\delta A} < \widehat{\alpha}_{\delta}) - \mathds{1}(e^{\delta A} < \alpha_{\delta})\right\}\right\rVert + \lVert \widehat{\alpha}_{\delta} - \alpha_{\delta}\rVert + \lVert \widehat{\gamma}_{\delta} - \gamma_{\delta}\rVert = o_P(1).
    \end{align*}

\noindent The conclusion follows by our consistency assumptions, and since for any $t > 0$,
    \begin{align*}
        &\left\lVert \mathds{1}(e^{\delta A} < \widehat{\alpha}_{\delta}) - \mathds{1}(e^{\delta A} < \alpha_{\delta})\right\rVert^2 \\
        &= P\left(|\mathds{1}(e^{\delta A} < \widehat{\alpha}_{\delta}) - \mathds{1}(e^{\delta A} < \alpha_{\delta})|\right) \\
        & \leq P\left(\mathds{1}(|e^{\delta A} - \alpha_{\delta}| \leq |\widehat{\alpha}_{\delta} - \alpha_{\delta}|)\right) \\
        & \leq P[|e^{\delta A} - \alpha_{\delta}| \leq t] + P[|\widehat{\alpha}_{\delta} - \alpha_{\delta}| > t] \\
        &\lesssim t^b + \frac{1}{t} \lVert \widehat{\alpha}_{\delta} - \alpha_{\delta}\rVert,
    \end{align*}
    we have $\left\lVert \mathds{1}(e^{\delta A} < \widehat{\alpha}_{\delta}) - \mathds{1}(e^{\delta A} < \alpha_{\delta})\right\rVert = o_P(1)$ by Lemma 6 of~\citet{levis2024}.
\end{proof}

\begin{proof}[Proof of Theorem~\ref{thm:est-cont-inc-wass}]
    This result follows identical reasoning to that in the proof of Theorem~\ref{thm:est-cont-inc}, specifically for the error decomposition of $\widehat{\tau}_{\delta}$. The only difference is that $\beta_{\delta}$ takes the place of $\gamma_{\delta}$ (i.e., substituting $A e^{\delta A}$ for $Y e^{\delta}$).
\end{proof}

\begin{proof}[Proof of Proposition~\ref{prop:von-Mises-bin}]
    The expansion for $\tau_{\delta}$ is simply Proposition~\ref{prop:von-Mises}, applied to the case where $A \in\{0,1\}$---alternatively, this is a restatement of Corollary 2 in~\citet{kennedy2019}.

    For $\chi_{\delta}$, observe that
    \begin{align*}
        & \chi_{\delta}(\overline{P}) - \chi_{\delta}(P) + \mathbb{E}_P(\dot{\chi}_{\delta}(O; \overline{P})) \\
        &= \mathbb{E}_P\left(\frac{\overline{\pi}(1 - \overline{\pi})}{e^{\delta} \overline{\pi} + (1 - \overline{\pi})} + \frac{\{(1 - \overline{\pi})^2 - e^{\delta}\overline{\pi}^2\}(\pi - \overline{\pi})}{\{e^{\delta}\overline{\pi} + (1 - \overline{\pi})\}^2} - \frac{\pi(1 - \pi)}{e^{\delta} \pi + (1 - \pi)}\right) \\
        &= \mathbb{E}_P\left(\frac{\overline{\pi}(1 - \overline{\pi})\{e^{\delta}\overline{\pi} + (1 - \overline{\pi})\} \{e^{\delta} \pi + (1 - \pi)\}}{\{e^{\delta}\overline{\pi} + (1 - \overline{\pi})\}^2 \{e^{\delta} \pi + (1 - \pi)\}}\right) \\
        & \quad \quad \quad \quad 
        + \mathbb{E}_P\left(\frac{\{(1 - \overline{\pi})^2 - e^{\delta}\overline{\pi}^2\}(\pi - \overline{\pi})\{e^{\delta} \pi + (1 - \pi)\}}{\{e^{\delta}\overline{\pi} + (1 - \overline{\pi})\}^2 \{e^{\delta} \pi + (1 - \pi)\}}\right) \\
        & \quad \quad \quad \quad
        - \mathbb{E}_P\left(\frac{\pi(1 - \pi)\{e^{\delta}\overline{\pi} + (1 - \overline{\pi})\}^2}{\{e^{\delta}\overline{\pi} + (1 - \overline{\pi})\}^2 \{e^{\delta} \pi + (1 - \pi)\}}\right) \\
        &= \mathbb{E}\left(\frac{e^{\delta}(\overline{\pi} - \pi)^2}{\{e^{\delta}\overline{\pi} + (1 - \overline{\pi})\}^2 \{e^{\delta} \pi + (1 - \pi)\}}\right),
    \end{align*}
    where the last equality comes from (tedious) expansion and simplification of the numerator.
\end{proof}

\begin{proof}[Proof of Theorem~\ref{thm:est-bin-inc}]
    The logic here is identical to that in the proof of Theorem~\ref{thm:est-cont-inc}. By Proposision~\ref{prop:von-Mises-bin}, the bias term in the error expansion for $\tau_{\delta}$ is
    \begin{align*}
        P\left(\{\widehat{\gamma}_{\delta} - \gamma_{\delta}\}\left\{1 - \frac{\alpha_{\delta}}{\widehat{\alpha}_{\delta}}\right\}\right)
        &= P\left(\left\{\frac{e^{\delta}\widehat{\pi} \widehat{\mu}_1 + (1 - \widehat{\pi})\widehat{\mu}_0}{\widehat{\alpha}_{\delta}} - \frac{e^{\delta}\pi \mu_1 + (1 - \pi)\mu_0}{\alpha_{\delta}}\right\}\left\{1 - \frac{\alpha_{\delta}}{\widehat{\alpha}_{\delta}}\right\}\right) \\
        &= O_P(\lVert \widehat{\pi} - \pi\rVert^2 + \lVert \widehat{\pi} - \pi\rVert\left\{\widehat{\mu}_0 - \mu_0\rVert + \lVert \widehat{\mu}_1 - \mu_1\right\}\rVert),
    \end{align*}
    by our boundedness assumptions and the fact that $\widehat{\alpha}_{\delta} - \alpha_{\delta} \equiv (e^{\delta} - 1)(\widehat{\pi} - \pi)$. That the bias term in the expansion for $\chi_{\delta}$ is $O_P(\lVert \widehat{\pi} - \pi\rVert^2)$ comes directly from the von Mises expansion in Proposition~\ref{prop:von-Mises-bin}. The empirical process term for $\tau_{\delta}$ is $o_P(n^{-1/2})$ by Lemma 2 in~\citet{kennedy2020b}, and the following calculation:
    \begin{align*}
        & \lVert \widehat{\Psi}_{\tau_{\delta}}  - \Psi_{\tau_{\delta}}\rVert \\
        &= \left\lVert \frac{e^{\delta} \widehat{\pi} \widehat{\phi}_1 + (1 - \widehat{\pi})\widehat{\phi}_0}{e^{\delta}\widehat{\pi} + (1 - \widehat{\pi})} +  \frac{e^{\delta}(\widehat{\mu}_1 - \widehat{\mu}_0)(A - \widehat{\pi})}{\{e^{\delta}\widehat{\pi} + (1 - \widehat{\pi})\}^2} - \frac{e^{\delta} \pi \phi_1 + (1 - \pi)\phi_0}{e^{\delta}\pi + (1 - \pi)} -  \frac{e^{\delta}(\mu_1 - \mu_0)(A - \pi)}{\{e^{\delta}\pi + (1 - \pi)\}^2}\right\rVert \\
        & \leq \left\lVert \frac{e^{\delta}(\widehat{\pi}\widehat{\phi}_1 - \pi \phi_1) + (1 - \widehat{\pi})\widehat{\phi}_0 - (1 - \pi)\phi_0}{\widehat{\alpha}_{\delta}}\right\rVert + \left\lVert (e^{\delta} \pi \phi_1 + (1 - \pi)\phi_0)\left(\frac{1}{\widehat{\alpha}_{\delta}} - \frac{1}{\alpha_{\delta}}\right) \right\rVert \\
        & \quad \quad + e^{\delta}\left\lVert \frac{(\widehat{\mu}_1 - \widehat{\mu}_0)(A - \widehat{\pi}) - (\mu_1 - \mu_0)(A - \pi)}{\widehat{\alpha}_{\delta}^2}\right\rVert + \left\lVert e^{\delta}(\mu_1 - \mu_0)(A - \pi)\left(\frac{1}{\widehat{\alpha}_{\delta}^2} - \frac{1}{\alpha_{\delta}^2}\right) \right\rVert \\
        & \lesssim \lVert \widehat{\pi}  - \pi\rVert + \lVert \widehat{\mu}_0  - \mu_0\rVert + \lVert \widehat{\mu}_1  - \mu_1\rVert \\
        &= o_P(1),
    \end{align*}
    by our consistency and boundedness assumptions. Similarly, the empirical process term in the error expansion for $\chi_{\delta}$ is also $o_P(n^{-1/2})$, as
    \begin{align*}
        & \lVert \widehat{\Psi}_{\chi_{\delta}}  - \Psi_{\chi_{\delta}}\rVert \\
        &= \left\lVert \frac{\widehat{\pi}(1 - \widehat{\pi})}{\widehat{\alpha}_{\delta}} + \frac{\{(1 - \widehat{\pi})^2 - e^{\delta}\widehat{\pi}^2\}(A - \widehat{\pi})}{\widehat{\alpha}_{\delta}^2} - \frac{\pi(1 - \pi)}{\alpha_{\delta}} + \frac{\{(1 - \pi)^2 - e^{\delta}\pi^2\}(A - \pi)}{\alpha_{\delta}^2}\right\rVert \\
        & \leq \left\lVert \frac{\widehat{\pi}(1 - \widehat{\pi}) - \pi(1 - \pi)}{\widehat{\alpha}_{\delta}} \right\rVert + \left\lVert \pi(1 - \pi)\left(\frac{1}{\widehat{\alpha}_{\delta}} - \frac{1}{\alpha_{\delta}}\right)\right\rVert\\
        & \quad \quad + \left\lVert \frac{\{(1 - \widehat{\pi})^2 - e^{\delta}\widehat{\pi}^2\}(A - \widehat{\pi}) - \{(1 - \pi)^2 - e^{\delta}\pi^2\}(A - \pi)}{\widehat{\alpha}_{\delta}^2}\right\rVert + \\
        & \quad \quad + \left\lVert \{(1 - \pi)^2 - e^{\delta}\pi^2\}(A - \pi)\left(\frac{1}{\widehat{\alpha}_{\delta}^2} - \frac{1}{\alpha_{\delta}^2}\right)\right\rVert\\
        & \lesssim \lVert \widehat{\pi} - \pi\rVert \\
        &= o_P(1),
    \end{align*}
    concluding the proof.
\end{proof}

Before proving Theorem~\ref{thm:est-bin-inc-prop}, we write out in detail the estimators for all four functionals $\zeta_{0}^-, \zeta_0^+, \zeta_1^-, \zeta_1^+$. In general, defining the functions $\kappa_{a}^{-} = \mathbb{E}((\gamma_a^{-}(X) - Y)_+ \mid X, A = a)$ and $\kappa_{a}^{+} = \mathbb{E}((Y - \gamma_a^{+}(X))_+ \mid X, A = a)$, our proposed estimators are $\widehat{\zeta}_a^{\pm} = \mathbb{P}_n[\widehat{\varphi}_{\delta, a}^{\pm}]$, where

\begin{align*}
    \varphi_{\delta,0}^- &= \frac{\pi}{e^{\delta}\pi + (1 - \pi)}\bigg\{\left(1 - \frac{1}{\Gamma}\right)\left\{(1 -\pi) \mu_0 + (1 - A) (Y - \mu_0)\right\} \\
    & \quad \quad \quad \quad \quad + \left(\Gamma - \frac{1}{\Gamma}\right)\left\{(1 - \pi)\left(\kappa_0^- - \frac{1}{1 + \Gamma} \gamma_0^-\right) - (1 - A)\left[(\gamma_0^- - Y)_+ - \kappa_0^-\right]\right\}\bigg\} \\
    & \quad \quad + \frac{(1 - \pi)^2 - e^{\delta}\pi^2}{\{e^{\delta}\pi + (1 - \pi)\}^2}(A - \pi)\left\{\left(1 - \frac{1}{\Gamma}\right)(\mu_0 - \gamma_0^-) + \left(\Gamma - \frac{1}{\Gamma}\right)\kappa_0^-\right\},
\end{align*}

\begin{align*}
    \varphi_{\delta,0}^+ &= \frac{\pi}{e^{\delta}\pi + (1 - \pi)}\bigg\{ -\left(1 - \frac{1}{\Gamma}\right)\left\{(1 - \pi) \mu_0 + (1 - A) (Y - \mu_0)\right\} \\
    & \quad \quad \quad \quad \quad + \left(\Gamma - \frac{1}{\Gamma}\right)\left\{(1 - \pi)\left(\frac{1}{1 + \Gamma} \gamma_0^+ + \kappa_0^+\right) + (1 - A)\left[(Y - \gamma_0^+)_+ - \kappa_0^+\right]\right\}\bigg\} \\
    & \quad \quad + \frac{(1 - \pi)^2 - e^{\delta}\pi^2}{\{e^{\delta}\pi + (1 - \pi)\}^2}(A - \pi)\left\{\left(1 - \frac{1}{\Gamma}\right)(\gamma_0^+ - \mu_0) + \left(\Gamma - \frac{1}{\Gamma}\right)\kappa_0^+\right\},
\end{align*}

\begin{align*}
    \varphi_{\delta,1}^- &= \frac{1 - \pi}{e^{\delta}\pi + (1 - \pi)}\bigg\{\left(1 - \frac{1}{\Gamma}\right)\left\{\pi \mu_1 + A (Y - \mu_1)\right\} \\
    & \quad \quad \quad \quad \quad + \left(\Gamma - \frac{1}{\Gamma}\right)\left\{\pi\left(\kappa_1^- - \frac{1}{1 + \Gamma} \gamma_1^-\right) - A\left[(\gamma_1^- - Y)_+ - \kappa_1^-\right]\right\}\bigg\} \\
    & \quad \quad + \frac{(1 - \pi)^2 - e^{\delta}\pi^2}{\{e^{\delta}\pi + (1 - \pi)\}^2}(A - \pi)\left\{\left(1 - \frac{1}{\Gamma}\right)(\mu_1 - \gamma_1^-) + \left(\Gamma - \frac{1}{\Gamma}\right)\kappa_1^-\right\},
\end{align*}

\begin{align*}
    \varphi_{\delta,1}^+ &= \frac{1 - \pi}{e^{\delta}\pi + (1 - \pi)}\bigg\{ -\left(1 - \frac{1}{\Gamma}\right)\left\{\pi \mu_1 + A (Y - \mu_1)\right\} \\
    & \quad \quad \quad \quad \quad + \left(\Gamma - \frac{1}{\Gamma}\right)\left\{\pi\left(\frac{1}{1 + \Gamma} \gamma_1^+ + \kappa_1^+\right) + A\left[(Y - \gamma_1^+)_+ - \kappa_1^+\right]\right\}\bigg\} \\
    & \quad \quad + \frac{(1 - \pi)^2 - e^{\delta}\pi^2}{\{e^{\delta}\pi + (1 - \pi)\}^2}(A - \pi)\left\{\left(1 - \frac{1}{\Gamma}\right)(\gamma_1^+ - \mu_1) + \left(\Gamma - \frac{1}{\Gamma}\right)\kappa_1^+\right\}.
\end{align*}

\noindent The asymptotic properties of each of these estimators are argued in exactly the same way. Therefore, we only prove a result for $\zeta_{\delta,1}^+$, as stated in Theorem~\ref{thm:est-bin-inc-prop}.

\begin{proof}[Proof of Theorem~\ref{thm:est-bin-inc-prop}]
As for the other results, we use the expansion
    \begin{align*}
        \widehat{\zeta}_{\delta,1}^+ - \zeta_{\delta,1}^+ &= \mathbb{P}_n[\widehat{\varphi}_{\delta,1}^+] - P(\varphi_{\delta,1}^+) \\
        &= (\mathbb{P}_n - P)\varphi_{\delta,1}^+ + P(\widehat{\varphi}_{\delta,1}^+ - \varphi_{\delta,1}^+) + (\mathbb{P}_n - P)(\widehat{\varphi}_{\delta,1}^+ - \varphi_{\delta,1}^+).
    \end{align*}
    Beginning with the second (asymptotic bias) term,
\begin{align*}
    &P(\widehat{\varphi}_{\delta,1}^+ - \varphi_{\delta,1}^+) \\
    &= P\bigg( \frac{1 - \widehat{\pi}}{e^{\delta}\widehat{\pi} + (1 - \widehat{\pi})}\bigg\{ -\left(1 - \frac{1}{\Gamma}\right)\left\{\widehat{\pi} \widehat{\mu}_1 + A (Y - \widehat{\mu}_1)\right\} \\
    & \quad \quad \quad \quad \quad + \left(\Gamma - \frac{1}{\Gamma}\right)\left\{\widehat{\pi}\left(\frac{1}{1 + \Gamma} \widehat{\gamma}_1^+ + \widehat{\kappa}_1^+\right) + A\left[(Y - \widehat{\gamma}_1^+)_+ - \widehat{\kappa}_1^+\right]\right\}\bigg\} \\
    & \quad \quad + \frac{(1 - \widehat{\pi})^2 - e^{\delta}\widehat{\pi}^2}{\{e^{\delta}\widehat{\pi} + (1 - \widehat{\pi})\}^2}(A - \widehat{\pi})\left\{\left(1 - \frac{1}{\Gamma}\right)(\widehat{\gamma}_1^+ - \widehat{\mu}_1) + \left(\Gamma - \frac{1}{\Gamma}\right)\widehat{\kappa}_1^+\right\}
    \\
    & \quad \quad \quad \quad - \left\{\left(1 - \frac{1}{\Gamma}\right)(\gamma_1^+ - \mu_1) + \left(\Gamma - \frac{1}{\Gamma}\right)\kappa_{1}^+\right\}\frac{\pi (1 - \pi)}{e^{\delta}\pi + (1 - \pi)}
    \bigg).
\end{align*}
By iterated expectations, and rearranging, this equals
\begin{align*}
     & P\bigg( \frac{1 - \widehat{\pi}}{e^{\delta}\widehat{\pi} + (1 - \widehat{\pi})}\bigg\{ -\left(1 - \frac{1}{\Gamma}\right)\left\{\widehat{\pi} \widehat{\mu}_1 + \pi (\mu_1 - \widehat{\mu}_1)\right\} \\
    & \quad \quad \quad \quad + \left(\Gamma - \frac{1}{\Gamma}\right)\left\{\widehat{\pi}\left(\frac{1}{1 + \Gamma} \widehat{\gamma}_1^+ + \widehat{\kappa}_1^+\right) + A\left[(Y - \widehat{\gamma}_1^+)_+ - (Y - \gamma_1^+)_+\right] + \pi\left[\kappa_1^+ - \widehat{\kappa}_1^+\right]\right\}\bigg\} \\
    & \quad \quad + \frac{(1 - \widehat{\pi})^2 - e^{\delta}\widehat{\pi}^2}{\{e^{\delta}\widehat{\pi} + (1 - \widehat{\pi})\}^2}(\pi - \widehat{\pi})\left\{\left(1 - \frac{1}{\Gamma}\right)(\widehat{\gamma}_1^+ - \widehat{\mu}_1) + \left(\Gamma - \frac{1}{\Gamma}\right)\widehat{\kappa}_1^+\right\}
    \\
    & \quad \quad - \frac{\pi (1 - \pi)}{e^{\delta}\pi + (1 - \pi)}\left\{\left(1 - \frac{1}{\Gamma}\right)(\gamma_1^+ - \mu_1) + \left(\Gamma - \frac{1}{\Gamma}\right)\kappa_{1}^+\right\}
    \bigg) \\
    &= P\bigg(\frac{1 - \widehat{\pi}}{e^{\delta}\widehat{\pi} + (1 - \widehat{\pi})}\bigg\{ -\left(1 - \frac{1}{\Gamma}\right)(\widehat{\pi} - \pi) (\mu_1 - \widehat{\mu}_1) \\
    & \quad \quad \quad \quad + \left(\Gamma - \frac{1}{\Gamma}\right)(\widehat{\pi} - \pi)\left(\frac{1}{1 + \Gamma} (\widehat{\gamma}_1^+ -\gamma_1^+) + (\widehat{\kappa}_1^+ - \kappa_1^+)\right) \\
    & \quad \quad \quad \quad
    + \left(\Gamma - \frac{1}{\Gamma}\right) \pi \left(\frac{1}{1 + \Gamma}\left(\widehat{\gamma}_1^+ - \gamma_1^+\right) + P\left((Y - \widehat{\gamma}_1^+)_+ - (Y > \gamma_1^+)_+ \mid X, A = 1\right)\right)\bigg\} \\
    & \quad \quad + \frac{(1 - \widehat{\pi})^2 - e^{\delta}\widehat{\pi}^2}{\{e^{\delta}\widehat{\pi} + (1 - \widehat{\pi})\}^2}(\pi - \widehat{\pi})\left\{\left(1 - \frac{1}{\Gamma}\right)((\widehat{\gamma}_1^+-\gamma_1^+) - (\widehat{\mu}_1 - \mu_1)) + \left(\Gamma - \frac{1}{\Gamma}\right)(\widehat{\kappa}_1^+ - \kappa_1^+)\right\}
    \\
    & \quad \quad +\left( \frac{1 - \widehat{\pi}}{e^{\delta}\widehat{\pi} + (1 - \widehat{\pi})} + \frac{(1 - \widehat{\pi})^2 - e^{\delta}\widehat{\pi}^2}{\{e^{\delta}\widehat{\pi} + (1 - \widehat{\pi})\}^2}(\pi - \widehat{\pi}) - \frac{\pi (1 - \pi)}{e^{\delta}\pi + (1 - \pi)}\right) \\
    &\quad \quad \quad \quad \times \left\{\left(1 - \frac{1}{\Gamma}\right)(\gamma_1^+ - \mu_1) + \left(\Gamma - \frac{1}{\Gamma}\right)\kappa_{1}^+\right\}
    \bigg).
\end{align*}
By Lemma~\ref{lemma:cvar} below, and the conditional second-order remainder calculation for $\chi_{\delta}$ (as in the proof of Proposition~\ref{prop:von-Mises-bin}), this bias term is $O_P(R_{\delta,1}^+)$. Finally, it is straightforward to verify that the empirical process term $(\mathbb{P}_n - P)(\widehat{\varphi}_{\delta,1}^+ - \varphi_{\delta,1}^+)$ is $o_P(n^{-1/2})$, by Lemma 2 in~\citet{kennedy2020b}, as $\lVert \widehat{\varphi}_{\delta,1}^+ - \varphi_{\delta, 1}^+\rVert = o_P(1)$ under our assumptions.
\end{proof}

\begin{lemma}\label{lemma:cvar}
    Under Assumption~\ref{ass:margin-2},
    \[\frac{1}{1 + \Gamma}\left(\widehat{\gamma}_1^+ - \gamma_1^+\right) + P\left((Y - \widehat{\gamma}_1^+)_+ - (Y > \gamma_1^+)_+ \mid X, A = 1\right) \lesssim |\widehat{\gamma}_1^+ - \gamma_1^+|^{1 + b}.\]
\end{lemma}

\begin{proof}
    Observe that
    \begin{align*}
        & P\left((Y - \widehat{\gamma}_1^+)_+ - (Y > \gamma_1^+)_+ \mid X, A = 1\right) \\
        &= P((Y - \widehat{\gamma}_1^+)\mathds{1}(Y > \widehat{\gamma}_1^+) - (Y - \gamma_1^+)\mathds{1}(Y > \gamma_1^+)\mid X, A = 1) \\
        &= P((Y - \widehat{\gamma}_1^+)\{\mathds{1}(Y > \widehat{\gamma}_1^+) - \mathds{1}(Y > \gamma_1^+)\} \mid X, A = 1) - (\widehat{\gamma}_1^+ - \gamma_1^+)P[Y > \gamma_1^+ \mid X, A = 1] \\
        &= P((Y - \widehat{\gamma}_1^+)\{\mathds{1}(Y > \widehat{\gamma}_1^+) - \mathds{1}(Y > \gamma_1^+)\} \mid X, A = 1) - \frac{1}{1 + \Gamma}(\widehat{\gamma}_1^+ - \gamma_1^+),
    \end{align*}
    and moreover that
    \begin{align*}
        &\left|P((Y - \widehat{\gamma}_1^+)\{\mathds{1}(Y > \widehat{\gamma}_1^+) - \mathds{1}(Y > \gamma_1^+)\} \mid X, A = 1)\right| \\
        &\leq P\big(\mathds{1}(\max{\{|Y - \gamma_1^+|,|Y - \widehat{\gamma}_1^+|\}} \leq |\widehat{\gamma}_1^+ - \gamma_1^+|) \cdot |Y - \widehat{\gamma}_1^+| \, \big| \, X, A = 1\big) \\
        & \leq |\widehat{\gamma}_1^+ - \gamma_1^+| P[|Y - \gamma_1^+| \leq |\widehat{\gamma}_1^+ - \gamma_1^+| \mid X, A = 1] \leq C \cdot |\widehat{\gamma}_1^+ - \gamma_1^+|^{1 + b},
    \end{align*}
    where the first inequality follows from Lemma 4 in~\citet{levis2024} or Lemma 1 in~\citet{kennedy2020b}, and the third inequality by Assumption~\ref{ass:margin-2}.
\end{proof}

\end{appendices}

\end{document}